\ifdefined\IsLLnCS
\documentclass{llncs}
\usepackage{amsfonts,amsmath,amssymb,boxedminipage,color,url,nccmath}
\else
\documentclass[11pt]{article}
\usepackage{amsfonts,amsmath,amssymb,boxedminipage,color,url,fullpage,nccmath,amsthm}
\fi

\usepackage{color}

\definecolor{weborange}{rgb}{.8,.3,.3}
\definecolor{webblue}{rgb}{0,0,.8}
\definecolor{internallinkcolor}{rgb}{0,.5,0}
\definecolor{externallinkcolor}{rgb}{0,0,.5}

\usepackage[pagebackref,
hyperfootnotes=false,
colorlinks=true,
urlcolor=externallinkcolor,
linkcolor=internallinkcolor,
filecolor=externallinkcolor,
citecolor=internallinkcolor,
breaklinks=true,
pdfstartview=FitH,
pdfpagelayout=OneColumn]{hyperref}

\usepackage{enumerate,paralist}
\usepackage[numbers]{natbib} 
\usepackage[labelfont=bf]{caption}
\usepackage{aliascnt}
\usepackage{cleveref}
\usepackage{xspace}

\usepackage{relsize}
\usepackage{graphics}
\ifdefined\IsLLnCS
\fi

\usepackage{amsmath}
\usepackage{amsthm}
\usepackage{amstext}
\usepackage{amsfonts}
\usepackage{amssymb}
\usepackage{cancel}
\usepackage{xspace}
\usepackage{xcolor}
\usepackage{indentfirst}
\usepackage{nicefrac}

\providecommand{\remove}[1]{}

\newcommand{\Draft}[1]{\ifdefined\IsDraft \texttt{ #1} \fi}

\newcommand{\TLLNCS}[2]{\ifdefined\IsLLNCS#1\else #2 \fi}

\ifdefined\IsDraft
\newcommand{\authnote}[2]{{\bf [{\color{red} #1's Note:} {\color{blue} #2}]}}

\else
\newcommand{\authnote}[2]{}
\fi

\newcommand{\sdotfill}{\textcolor[rgb]{0.8,0.8,0.8}{\dotfill}} 

\newenvironment{protocol}{\begin{proto}}{\vspace{-\topsep}\sdotfill\end{proto}}

\newcommand{\ie} {i.e.,\ }
\newcommand{\eg} {e.g.,\ }
\newcommand{\wrt} {with respect to\ }

\newcommand{\cf}{{cf.,\ }}

\newcommand{\ceil}[1]{\left\lceil #1 \right\rceil}

\newcommand{\set}[1]{\ens{#1}}

\newcommand{\half}{\tfrac{1}{2}}

\newcommand{\eqdef}{:=}

\newcommand{\N}{{\mathbb{N}}}

\newcommand{\F}{{\mathbb F}}

\newcommand{\zo}{\set{0,1}}

\newcommand{\zs}{\zo^\ast}

\newcommand{\suchthat}{{\;\; : \;\;}}

\newcommand{\al}{\alpha}
\newcommand{\be}{\beta}
\newcommand{\eps}{\epsilon}

\newcommand{\from}{\leftarrow}

\newcommand{\Exp}{\operatorname*{E}}

\newcommand{\Supp}{\operatorname{Supp}}

\newcommand{\Image}{\operatorname{Im}}

\newcommand{\MathFam}[1]{\mathcal{#1}}

\newcommand{\FFam}{\MathFam{F}}

\newcommand{\MathAlg}[1]{\mathsf{#1}\xspace}

\renewcommand{\cref}{\Cref}

\TLLNCS{
	\newaliascnt{claiml}{theorem}
	\newtheorem{claiml}[claiml]{Claim}
	\aliascntresetthe{claiml}
	
	\renewenvironment{claim}{\begin{claiml}}{\end{claiml}}
	
	\newaliascnt{lemmal}{theorem}
	\newtheorem{lemmal}[lemmal]{Lemma}
	\aliascntresetthe{lemmal}
	\renewenvironment{lemma}{\begin{lemmal}}{\end{lemmal}}
	
	\newaliascnt{propositionl}{theorem}
	\newtheorem{propositionl}[propositionl]{Proposition}
	\aliascntresetthe{propositionl}
	\renewenvironment{proposition}{\begin{propositionl}}{\end{propositionl}}

	\newaliascnt{definitionl}{theorem}
	\newtheorem{definitionl}[definitionl]{Definition}
	\aliascntresetthe{definitionl}
	\renewenvironment{definition}{\begin{definitionl}}{\end{definitionl}}
	
	\newaliascnt{corollaryl}{theorem}
	\newtheorem{corollaryl}[definitionl]{Corollary}
	\aliascntresetthe{corollaryl}
	\renewenvironment{corollary}{\begin{corollaryl}}{\end{corollaryl}}

}
{
	\newtheorem{theorem}{Theorem}[section]

	\newaliascnt{lemma}{theorem}
	\newtheorem{lemma}[lemma]{Lemma}
	\aliascntresetthe{lemma}
	
	\newaliascnt{claim}{theorem}
	\newtheorem{claim}[claim]{Claim}
	\aliascntresetthe{claim}

	\newaliascnt{corollary}{theorem}
	\newtheorem{corollary}[corollary]{Corollary}
	\aliascntresetthe{corollary}
	
	\newaliascnt{proposition}{theorem}
	
	\aliascntresetthe{proposition}

	\newaliascnt{conjecture}{theorem}
	
	\aliascntresetthe{conjecture}

	\newaliascnt{definition}{theorem}
	\newtheorem{definition}[definition]{Definition}
	\aliascntresetthe{definition}

	\newaliascnt{remark}{theorem}
	\newtheorem{remark}[remark]{Remark}
	\aliascntresetthe{remark}

	\newaliascnt{example}{theorem}
	
	\aliascntresetthe{example}
}

\crefname{lemma}{Lemma}{Lemmas}
\crefname{figure}{Figure}{Figures}
\crefname{claim}{Claim}{Claims}
\crefname{corollary}{Corollary}{Corollaries}
\crefname{proposition}{Proposition}{Propositions}
\crefname{conjecture}{Conjecture}{Conjectures}
\crefname{definition}{Definition}{Definitions}
\crefname{remark}{Remark}{Remarks}
\crefname{exmaple}{Example}{Examples}

\newaliascnt{construction}{theorem}

\aliascntresetthe{construction}
\crefname{construction}{Construction}{Constructions}

\newaliascnt{fact}{theorem}
\newtheorem{fact}[fact]{Fact}
\aliascntresetthe{fact}
\crefname{fact}{Fact}{Facts}

\newaliascnt{notation}{theorem}
\newtheorem{notation}[notation]{Notation}
\aliascntresetthe{notation}
\crefname{notation}{Notation}{Notation}

\crefname{equation}{Equation}{Equations}

\newaliascnt{proto}{theorem}

\newtheorem{proto}[proto]{Protocol}

\aliascntresetthe{proto}
\crefname{proto}{protocol}{protocols}

\newaliascnt{algo}{theorem}
\newtheorem{algo}[algo]{Algorithm}
\aliascntresetthe{algo}
\crefname{algo}{algorithm}{algorithms}

\newaliascnt{expr}{theorem}
\newtheorem{expr}[expr]{Experiment}
\aliascntresetthe{expr}
\crefname{experiment}{experiment}{experiments}

\newcommand{\stepref}[1]{Step~\ref{#1}}

\def\FullBox{$\Box$}
\def\qed{\ifmmode\qquad\FullBox\else{\unskip\nobreak\hfil
\penalty50\hskip1em\null\nobreak\hfil\FullBox
\parfillskip=0pt\finalhyphendemerits=0\endgraf}\fi}

\def\qedsketch{\ifmmode\Box\else{\unskip\nobreak\hfil
\penalty50\hskip1em\null\nobreak\hfil$\Box$
\parfillskip=0pt\finalhyphendemerits=0\endgraf}\fi}

\newcommand{\ex}[2]{\Exp_{#1}\left[#2\right]}

\newcommand{\pr}[1]{\Pr\left[#1\right]}
\newcommand{\ppr}[2]{\Pr_{#1}\left[#2\right]}

\newcommand{\ens}[1]{\left\{#1\right\}}
\newcommand{\size}[1]{\left|#1\right|}

\newcommand{\cs}{{\cal{S}}}

\newcommand{\cE}{\mathcal{E}}
\newcommand{\cX}{\mathcal{X}}
\newcommand{\cY}{\mathcal{Y}}
\newcommand{\cF}{\mathcal{F}}
\newcommand{\cG}{{\mathcal{G}}}

\newcommand{\Ac}{\MathAlgX{A}}
\newcommand{\Bc}{\MathAlgX{B}}

\newcommand{\Pc}{\MathAlgX{P}}

\newcommand{\Tableofcontents}{
	\ifdefined\IsLLNCS \else
	\thispagestyle{empty}
	\pagenumbering{gobble}
	\clearpage
	\ifdefined\IsSubmission \else
	\setcounter{tocdepth}{2}
	\tableofcontents
	\thispagestyle{empty}
	\clearpage
	\fi
	\pagenumbering{arabic}
	\fi
}

\newcommand{\hide}[1]{ }
\ifdefined\IsLLNCS
\else

\fi

\newcommand{\rv}[1]{\mathrm{#1}}

\newcommand{\nat}{\mathbb{N}}

\newcommand{\setx}{{\cal{X}}}
\newcommand{\sety}{{\cal{Y}}}
\DeclareMathAlphabet \mathbfcal{OMS}{cmsy}{b}{n}

\newcommand{\setY}{{\rv{Y}}}

\newcommand{\Ensuremath}[1]{\ensuremath{#1}\xspace}

\newcommand{\MathAlgX}[1]{\Ensuremath{\MathAlg{#1}}}

\newcommand{\Th}[1]{#1\ensuremath{^{\rm th}}}
\newcommand{\ith}{\Th{i}\xspace}
\newcommand{\jth}{\Th{j}\xspace}

\newcommand{\nfrac}[2]{\nicefrac{#1}{#2}}

\newcommand{\Cuv}{\cE}

\newcommand{\fld}{\ensuremath{\mathbb{F}}\xspace}

\newcommand{\str}[1]{\ensuremath{{\zo}^{#1}\xspace}}

\newcommand{\inrprd}[1]{\langle #1 \rangle}

\newcommand{\cK}{\ensuremath{\mathcal{K}}}
\newcommand{\cS}{\ensuremath{\mathcal{S}}}
\newcommand{\cT}{\ensuremath{\mathcal{T}}}
\newcommand{\cL}{\ensuremath{\mathcal{L}}}
\newcommand{\cQ}{\ensuremath{\mathcal{Q}}}

\newcommand{\func}[1]{ \qopname\relax o{#1} }

\newcommand{\rank}{\func{rank}}
\newcommand{\Span}{\func{Span}}

\newcommand{\stack}[2]{ \begin{pmatrix} #1 \\ #2 \end{pmatrix} }

\newcommand{\numqry}{q}

\newcommand{\maxi}{ i }

\newcommand{\Zto}[1]{\ensuremath{Z_{#1}}}

\newcommand{\fspace}{ \ensuremath{ {\FFam_n} }}
\newcommand{\getf}{ \ensuremath{ f \from \fspace } }

\newcommand{\matl}{ \fld^{\ell \times n} }

\newcommand{\matd}{ \fld^{\numqry \times n} }
\newcommand{\vecl}{ \fld^{\ell} }

\newcommand{ \pA }{\Ac}
\newcommand{ \pB }{\Bc}
\newcommand{\true}{\text{True}}
\newcommand{\false}{\text{False}}

\newcommand{\Inv}{\MathAlgX{C}}
\newcommand{\preprocessor}{\Inv_{\mathsf{pre}}}
\newcommand{\queries}{\Inv_\mathsf{qry}}
\newcommand{\decoder}{\Inv_{\mathsf{dec}}}
\newcommand{\triplet}{( \preprocessor, \queries, \decoder) }

\newcommand{\InvB}{\MathAlgX{D}}

\newcommand{\queriesB}{\InvB_\mathsf{qry}}
\newcommand{\decoderB}{\InvB_{\mathsf{dec}}}

\newcommand{\pref}[1]{\ensuremath{f^{-1}(#1)}}

\newcommand{\algOutputfa}[2]{ \decoder \left( #1, #2, f\left( \queries(#1, #2) \right)  \right) }

\newcommand{\mat}[1]{\textbf{#1}}

\newcommand{\mE}{\mat{E}}
\newcommand{\mA}{\mat{A}}
\newcommand{\mB}{\mat{B}}

\newcommand{\mM}{\mat{M}}

\newcommand{\Dnote}[1]{\authnote{Dror}{#1}}
\newcommand{\Nnote}[1]{\authnote{noam}{#1}}
\newcommand{\Inote}[1]{\authnote{Iftach}{#1}}

\newcommand{\Reviewer}[1]{\authnote{Reviewer}{#1}}

\title{Lower Bounds on  the Time/Memory Tradeoff of\\  Function Inversion
	\Draft{\\\small \sc Working Draft: Please Do not Distribute}
	\thanks{An extended abstract of this work appeared  in TCC 2020 \cite{chawin2020lower}}
}

 \ifdefined\IsAnon
 	\author{}
	 \date{}
 	\LLnCS{ \institute{}}
 \else
 	\ifdefined\IsLLnCS
 \else
    \author{Dror Chawin\thanks{School of Computer Science, Tel Aviv University. Emails: \texttt{quefumas@gmail.com,iftachh@cs.tau.ac.il,noammaz@gmail.com}. Research supported by ERC starting grant 638121 and Israel Science Foundation grant   666/19.}  
    	\and Iftach Haitner\footnotemark[2] \thanks{Member of the  Check Point Institute for Information Security. }
    	    	\and Noam Mazor\footnotemark[2]
    }
 \fi
 \fi

\begin{document}
\sloppy
\maketitle

\begin{abstract}
	
We study time/memory tradeoffs of \textit{function inversion}: an algorithm, \ie  an  \textit{inverter},  equipped with  an $s$-bit advice on a randomly chosen function $f\colon [n] \mapsto [n]$ and using $q$  oracle queries to $f$, tries to invert a randomly chosen output $y$ of $f$, \ie to find $x\in f^{-1}(y)$.  Much progress was done regarding  \textit{adaptive} function inversion---the inverter is allowed to make  \emph{adaptive}  oracle queries. \citeauthor{Hellman80} [IEEE transactions on Information Theory '80]  presented an adaptive  inverter   that inverts with high probability  a random $f$. \citeauthor{FiatN00s} [SICOMP '00] proved that  for any $s,q$ with $s^3 q = n^3$ (ignoring low-order  terms), an $s$-advice, $q$-query   variant of \citeauthor{Hellman80}'s  algorithm inverts a constant fraction of the image points of \emph{any} function.   \citeauthor{Yao90} [STOC '90] proved a lower bound of $sq\ge n$ for  this problem. Closing the gap between the above lower and upper bounds is a long-standing open question.

Very little is known for the \textit{non-adaptive} variant of the question---the inverter chooses  its queries \emph{in advance}. The only known upper bounds, \ie inverters, are the \emph{trivial} ones (with $s+q= n$), and the only lower bound is the above bound of  \citeauthor{Yao90}. In a recent work, \citeauthor{CorriganK19} [TCC '19] partially justified the difficulty of finding lower bounds on non-adaptive inverters, showing that  a lower bound on the time/memory tradeoff of  non-adaptive inverters implies a lower  bound on low-depth Boolean circuits. Bounds that,  for a strong  enough choice of parameters, are notoriously hard to prove. 

We make progress on the above intriguing  question, both for the adaptive and the non-adaptive case, proving the  following lower bounds on restricted families of  inverters:
\begin{description}
	\item[Linear-advice (adaptive inverter).] If the advice string is a linear function of $f$ (\eg $A\times f$, {for some matrix $A$,} viewing $f$ as  a vector in $[n]^n$), then $s+q \in \Omega(n)$. The bound generalizes to the  case where the advice string of $f_1 + f_2$, \ie the  coordinate-wise addition of the truth tables of $f_1$ and $f_2$, can be computed from the description of $f_1$ and $f_2$  by a \emph{low} communication protocol.
	
	\item[Affine non-adaptive decoders.] If the  non-adaptive inverter    has an \textit{affine decoder}---it  outputs  a linear function, determined by the advice string and the element to invert, of the query answers---then $s \in \Omega(n)$ (regardless of   $q$).
	
	\item [Affine non-adaptive decision trees.] If the  non-adaptive inversion algorithm is a $d$-depth \textit{affine decision tree}---it outputs the evaluation of a decision tree whose nodes compute  a linear function of the answers to the queries---and $q < cn$ for some universal $c>0$, then $s\in \Omega(n/d \log n)$.
\end{description}
 
\end{abstract}


\Tableofcontents

\section{Introduction}
In the \textit{function-inversion} problem, an algorithm,  \textit{inverter}, attempts to find a preimage for a randomly chosen $y\in[n]$ of a random function $f\colon [n] \to [n]$. The inverter is equipped with  an $s$-bit advice on $f$, and may  make $q$   oracle queries to $f$. Since $s$  lowerbounds  the inverter space complexity and $q$ lowerbounds the inverter time complexity, it is common to refer to the relation between  $s$ and $q$ as the inverter's  \textit{time/memory tradeoff}.  The function-inversion problem is central to both theoretical and  practical cryptography. On the  theoretical end, the security of many systems relies on the existence of one-way functions. While the task of inverting one-way functions  is very different from  that  of inverting random functions,  understanding the latter task is critical towards developing lower bounds on  the possible (black-box) implications of one-way functions, \eg \citet{ImpagliazzoRu89,GennaroGKT05}.   But advances on this problem (at least on the positive side, \ie inverters) are likely to find practical applications. Indeed,  algorithms for function inversion are used to expose weaknesses in  existing cryptosystems.

Much progress was done regarding   \emph{adaptive} function inversion---the inverter may choose its queries adaptively (\ie based on answers for previous queries).  \citet{Hellman80}   presented an adaptive  inverter   that inverts with high probability  a random $f$. \citet{FiatN00s} proved that  for any $s,q$ with $s^3 q = n^3$ (ignoring low-order  terms), an $s$-advice $q$-query   variant of \citeauthor{Hellman80}'s  algorithm inverts a constant fraction of the image points of \emph{any} function.  \citet{Yao90} proved a lower bound of $s\cdot q\ge n$ for this  problem.  Closing the gap between the above lower and upper bounds is a long-standing open question.  In contrast, very little is known about the non-adaptive variant of this problem---the inverter  performs all queries at once.  This variant is interesting since such inverter is likely  be highly  parallelizable, making it significantly more tractable for real world applications.  The only known upper bounds for this variant, \ie inverters, are the \emph{trivial} ones (\ie $s+q = n$), and the only known lower bound is the above bound of  \citet{Yao90}. In a recent work, \citet{CorriganK19} have partially justified the difficulty of finding lower bounds on this seemingly   easier to tackle problem, showing that  lower bounds on non-adaptive inversion imply  circuit  lower bounds that, for strong enough parameters, are notoriously hard (see further details in \cref{sec:intro:OurResult:Applications}).

\subsection{Our Results}\label{sec:intro:OurResult}
We make progress on this intriguing  question, proving lower bounds on restricted families of  inverters.   To state our results, we use the following formalization to capture  inverters with a preprocessing phase: such  inverters have  two parts, the \emph{preprocessing} algorithm that gets as input the function to invert $f$ and  outputs an advice string $a$, and the \emph{decoding} algorithm that takes as input the element to invert $y$, the advice string $a$, and using restricted query access to $f$ tries to find a preimage of $y$.  We start with describing our bound for the time/memory tradeoff of linear-advice (adaptive) inverters, and then present our lower bounds for non-adaptive inverters. In the following, fix $n\in \N$ and let $\FFam$ be the set of all functions from $[n]$ to $[n]$.

\subsubsection{Linear-advice Inverters}
We start with a more formal description of adaptive function inverters.
\begin{definition}[Adaptive inverters, informal]\label{def:intro,adaptiveInverterWithSuccess}
An {\sf $s$-advice, $\numqry$-query adaptive inverter} is a  deterministic algorithm pair $\Inv\eqdef(\preprocessor, \decoder)$, where $\preprocessor:\FFam \to \str{s}$, and $\decoder^{(\cdot)}:[n]\times\str{s}\to [n]$ is a $q$-query  algorithm. We say that {\sf $\Inv$ inverts $\FFam$ with high probability} if 
	
	\begin{align*}
	\ppr{\stackrel{f\gets \FFam}{a \eqdef \preprocessor(f)}}{ \ppr{\stackrel{x\gets[n]}{y\eqdef f(x)}}{\decoder^f(y,a) \in f^{-1}(y)}\ge 1/2}\ge 1/2. 
	\end{align*}
\end{definition}
It is common to refer to $a$ ($\eqdef \preprocessor(f)$) as the \textit{advice string}.
 In \textit{linear-advice} inverters, the preprocessing  algorithm $\preprocessor$ is restricted to output  a linear function of $f$. That is,   $\preprocessor(f_1) + \preprocessor(f_2) = \preprocessor(f_1+ f_2)$, where the addition $f_1+ f_2$ is coordinate-wise \wrt an  arbitrary group over $[n]$, and the addition $\preprocessor(f_1) + \preprocessor(f_2)$ is over an arbitrary group that contains the image of   $\preprocessor$. An example of such a preprocessing algorithm is $\preprocessor(f)  \eqdef A\times f$, for $A\in \zo^{s\times n}$, viewing $f\in \FFam$ as a vector in $[n]^n$. For such inverters, we present the following bound.

\begin{theorem}[Bound on linear-advice inverters]\label{thm:intro:LinearAdvice}
	Assume there exists an $s$-advice $\numqry$-query inverter with linear preprocessing that inverts  $\FFam$ with high probability. Then   $s + \numqry \cdot \log n \in \Omega(n)$.
\end{theorem}
	\remove{\Reviewer{maybe instead of $s$, bound with the $CC$ of computing $\preprocessor(f_A + f_B)$, to generalize to more preprocessors?}
	\Reviewer{\#3 agrees, suggests an example which I didn't understand}}

We prove \cref{thm:intro:LinearAdvice}  via a reduction from \textit{set disjointness}, a classical problem in the study of two-party communication complexity. The above result generalizes to the following bound that replaces the  restriction on the decoder (\eg linear and short output) with the ability to compute the advice string of $f_1 + f_2$  by a low-communication protocol over the inputs $f_1$ and $f_2$.
\begin{theorem}[Bound on additive-advice inverters, informal]\label{thm:intro:LinearAdviceGen}
Assume there exists a $\numqry$-query inverter $\Inv\eqdef(\preprocessor, \cdot)$ and an $s$-bit communication two-party protocol $(\Pc_1,\Pc_2)$ such that for every $f_1,f_2 \in \FFam$,  the output of $\Pc_1$ in $(\Pc_1(f_1),\Pc_2(f_2))$  equals with constant probability   to $\preprocessor(f_1 + f_2)$. Then   $s + \numqry \cdot \log n \in \Omega(n)$.
\end{theorem}
The above bound indeed generalizes \cref{thm:intro:LinearAdvice}:  a preprocessing algorithm of the type required by  \cref{thm:intro:LinearAdvice} immediately  implies  a two-party  protocol of the type required by \cref{thm:intro:LinearAdviceGen}.

\subsubsection{Non-adaptive Inverters}
In the non-adaptive setting, the decoding algorithm has two phases: the \textit{query selection} algorithm that chooses the queries as a function of the advice and the element to invert $y$, and the actual decoder that receives the answers to these queries along with the advice string and $y$.

\begin{definition}[Non-adaptive inverters, informal]\label{def:intro:nonAdaptInverters}
	An  {\sf $s$-advice, $q$-query   non-adaptive inverter} is a deterministic algorithm  triplet of the form $\Inv\eqdef \triplet$, where $
	\preprocessor\colon\FFam \to \str{s}$,
	$\queries\colon [n] \times \str{s} \to [n]^\numqry$ and
	$\decoder \colon[n] \times \str{s} \times [n]^\numqry \to [n]$. We say that {\sf $\Inv$ inverts $\FFam$ with high probability} if 
	
	\begin{align*}
	\ppr{\stackrel{f\gets \FFam}{a = \preprocessor(f)}}{ \ppr{\stackrel{x\gets[n]}{\stackrel{y=f(x)}{v = \queries(y,a)}}}{\decoder(y,a,f(v)) \in f^{-1}(y)}\ge 1/2}\ge 1/2. 
	\end{align*}
\end{definition}
Note that  the query vector $v$ is of length $q$, so the answer vector $f(v)$ contains $q$ answers.   Assuming there exists  a field $\F$ of size $n$ (see \cref{rem:fieldSize}),  we provide two lower bounds for such inverters.

\paragraph{Affine decoders.}
The first bound regards inverters with \emph{affine decoders}.   A decoder algorithm $\decoder$ is \emph{affine} if it computes an affine function of $f$'s answers. That is, for every  image $y\in [n]$ and advice   $a\in \zo^s$, there exists  a $\numqry$-sparse vector $\al_y^a\in\fld^n$ and a field element $\be_y^a\in\fld$ such that  $\decoder(y,a,f(\queries(y,a)))= \inrprd{\al_y^a, f}+ \be_y^a$ for every $f\in \FFam$. For this type of inverters, we present the following lower bound.

\begin{theorem}[Bound on non-adaptive  inverters with affine decoders, informal]\label{thm:intro:AffineDecoders}
Assume there exists an $s$-advice non-adaptive function  inverter with an affine decoder, that inverts $\FFam$ with high probability. Then $ s \in \Omega(n)$.
\end{theorem}
Note that the above bound on $s$ holds even if the inverter queries $f$ on all inputs. While \cref{thm:intro:AffineDecoders} is not very insightful for its own sake,  as we cannot expect a non-adaptive inverter to have such  a limiting structure, it is important since it can be  generalized to \emph{affine decision trees},  a much richer family of non-adaptive inverters defined below.  In addition, the result should be contrasted with the question  of \textit{black-box function computation}, see \cref{sec:intro:RelatedWork:IndexCoding}, for which   linear algorithm are  \emph{optimal}. Thus,  \cref{thm:intro:AffineDecoders} highlights the differences between these two related problems.

\remove{\Reviewer{\#1 interprets Affine decoders as having a single ``affine query'' to $f$, rather than $q$ individual queries, so this theorem suggests that affine queries are not stronger than normal single queries}}
\remove{\Reviewer{\#2 wants more justification for addressing these classes of inverters, and mentions they are very weak since they don't even include a trivial inverter which simply looks at the query results and sees if the required pre-image is among them. Dror: actually, we did address this kind of trivial inverter in earlier drafts, and obtained bounds that are dependent on $q$. Perhaps we should consider addressing this issue?}}
\remove{\Reviewer{\#4 would like to see more discussion of our choice of models, and why linear preprocessing admits adaptive queries, while linear decoding doesn't}}
\remove{\Reviewer{\#3 would like us to address possible (even trivial) upper bounds for these models. I'm not sure we can come up with one}}

\paragraph{Affine decision trees.}
The second  bound regards inverters whose decoders are  \emph{affine decision trees}.  An \emph{affine decision tree} is a decision tree  whose nodes compute  an \emph{affine} function, over $\F$,  of the input vector.  A decoder algorithm $\decoder$ is an  \emph{affine decision tree},  if for  every  image $y\in [n]$, advice   $a\in \zo^s$ and queries $v = \queries(y,a)$, there exists an affine decision tree $\cT^{y,a}$   such that $\decoder(y,a,f(v))= \cT^{y,a}(f)$ (\ie the output of $\cT^{y,a}$ on input $f$) for every $f\in \FFam$. For such inverters, we present the following bound.

\begin{theorem}[Bounds on non-adaptive  inverters with affine decision-tree decoders]\label{thm:intro:AffineTree}
	Assume there exists an $s$-advice $q$-query   non-adaptive function  inverter with a $d$-depth affine decision-tree decoder, that inverts $\FFam$ with high probability. Then  the following hold:
	\begin{itemize}
		\item 	$q < c  n$, for some universal constant  $c$, $\implies $ $s \in \Omega(n/d\log n)$.
		
		\item 	$q \in  n^{1- \Theta(1)}$ $\implies $  $ s \in \Omega( n/d)$.
	\end{itemize}
\end{theorem}
That is, we pay  a factor of $1/d$ comparing to the affine decoder bound, and the bound on $s$ only holds for not too large $q$. Affine decision trees  are much stronger than affine decoders, since the choice of the affine operations  it computes can be \emph{adaptively dependent} on the results of previous affine operations.   For example, a depth $d$  affine decision tree can compute \emph{any} function on $d$ linear combinations of the inputs. In particular,  multiplication of function values, \eg $f(1)\cdot f(2)$, which cannot be computed by an affine decoder, can be computed by a depth two decision tree. We note that an affine decision tree of depth $q$ can compute \emph{any} function of its $q$ queries.  Unfortunately, for $d=q$, our bound only  reproduces (up to log factors) the lower bound of \citet{Yao90}.

\remove{
\Reviewer{\#3: Two of your results focus on non-adaptive inverters. Does it seem possible to prove better-than-Yao lower bounds against partially adaptive inverters? For example, inverters that get to make their queries in a constant number of rounds?}
\Dnote{ Yao is (quite) tight for the adaptive case, and improving on Yao for the non-adaptive case implies new circuit bounds.
Improving on Yao at any point on the adaptiveness spectrum is likely to imply new bounds for either of the extremes, which is very unlikely.
}
\Reviewer{\#3 would be interested in seeing explicit discussion of the tightness of the bounds}

\Reviewer{\#4 ``It would also be nice if you could (informally) explain why the trivial counting argument/Yao's lower bound does not give any new results for the case when one of the algorithms is linear.''}
\Dnote{If we can generalize Yao's argument to accept 1 or $d$ ``affine queries'', instead of $q$ individual queries, perhaps that would imply our lower bound as well, in a much simpler way?}
}

\begin{remark}[Field size]\label{rem:fieldSize}
	In \cref{thm:intro:AffineDecoders,thm:intro:AffineTree}, the field size is assumed to be exactly $n$ (the domain of the function to invert). Decoders over fields smaller than $n$ are not particularly useful, since their output cannot cover all possible preimages of $f$. Our proof breaks down for fields of size larger than $n$, since we cannot use linear equations to represent the constraint that the decoder's output must be contained in the smaller set $[n]$. 
\end{remark}

\subsubsection{Applications to Valiant's Common-bit Model}\label{sec:intro:OurResult:Applications}
\citet{CorriganK19} showed that  a lower bound on the time/memory tradeoff of  \textit{strongly non-adaptive} function inverters---the queries may not depend  on the advice---implies a lower  bound on  circuit size in  \textit{Valiant's common-bit model} \cite{valiant1977graph,valiant1992boolean}.  Applying the reduction  of \cite{CorriganK19} with \cref{thm:intro:AffineTree} yields the following  bound:   for every $n\in \N$ for which  there exits an $n$-size field $\fld$, there is an explicit function $f\colon\fld^n\mapsto\fld^n$ that  cannot be computed by a three-layer circuit of the following structure:
\begin{enumerate}
	\item It has  $o(\nfrac n {d\log n})$ middle layer gates.

	\item  Each  output gate  is connected to  $n^{1-\Theta(1)}$ inputs gates (and to an arbitrary number of middle-layer gates).
	
	\item  Each  output gate computes  a function which is computable by a $d$-depth affine decision tree in the inputs (and depends arbitrarily on the middle layer).
	
\end{enumerate}
In fact, our bound yields that such circuits cannot even approximate $f$ so that every output gate outputs the right value with probability larger than $1/2$, over the inputs. \Inote{it it really $1/2$ or smaller $c>0$} \Nnote{ It can be every constant (but we used 1/2 when stated our main result)}

\subsection{Additional Related Work}\label{sec:intro:RelatedWork}

\subsubsection{Adaptive Inverters}

\paragraph{Upper bounds.}
The main result in this setting is the $s$-advice, $q$-query inverter  of \citet{Hellman80,FiatN00s} that inverts a constant fraction of the image of any function, for any $s,q$ such that $s^3q = n^3$ (ignoring low-order terms). When used for random \emph{permutations}, a variant on the same idea implies an  optimum inverter with  $s\cdot q = n$. The inverter  of \citeauthor{Hellman80,FiatN00s} has found application  to practical cryptanalysis, \eg \citet{biryukov2000cryptanalytic,biryukov2000real,oechslin2003making}.

\paragraph{Lower bounds.}

A long line of research (\citet{GennaroGKT05,dodis2017fixing,abusalah2017beyond,unruh2007random,coretti2018random,de2010time}) provides lower bounds for various variations on the classical setting, such as that of randomized inversion algorithms that succeed on a sub-constant fraction of functions. None of these lower bounds, however,  manage to improve on \citeauthor{Yao90}'s lower bound of $s\cdot q = n$, leaving a large  gap between this lower bound and \citeauthor{Hellman80,FiatN00s}'s inverter.

\subsubsection{Non-adaptive Inverters}
\paragraph{Upper bounds.}
In contrast with the adaptive case, it is not clear how to exploit  non-adaptive queries in a non trivial way.  Indeed, the only known inverters are the  trivial ones (roughly, the advice is the   function description, or the inverter queries the function on all inputs).

\paragraph{Lower bounds.}
Somewhat surprisingly, the  only known lower bound for non-adaptive inverters is \citeauthor{Yao90}'s, mentioned above. This defies the basic intuition that this task should be easier than the adaptive case, due to the extreme limitations under which  non-adaptive inverters operate. This difficulty was partially justified by the recent reduction of \citet{CorriganK19} (see  \cref{sec:intro:OurResult:Applications}) that implies  that a  strong enough lower bound  on even strongly non-adaptive inverters, would yield a lower bound on low-depth Boolean circuits that is   notoriously hard to prove. 

\subsubsection{Relation to Data Structures}
The problem of function inversion with advice may also be phrased as a problem in data structures, where the advice string serves as a succinct data structure for answering questions about $f$. In particular, it bears strong similarity to the {\em substring search }problem using the cell-probe model \cite{yao1981should}. This is the task of ascertaining the existence of a certain element within a large, unsorted database, using as few queries to the database and as little preprocessing as possible. Upper and lower bounds easily carry over between the two problems, a connection which was made in \citet{CorriganK19}, where it was used to obtain previously unknown upper bounds on substring search.

\remove{\Reviewer{\#3 would feel better if we mention here some relevant lower/upper bounds on data structures.
		Mentions in particular: Gál, Anna, and Peter Bro Miltersen 2007 "The cell probe complexity of succinctdata structures.",
		and Brody, Joshua, and Kasper Green Larsen 2012 "Adapt or die: Polynomial lower bounds for non-adaptive dynamic data structures." 
}}

\subsubsection{Index Coding and Black-box Function Computation}\label{sec:intro:RelatedWork:IndexCoding}
A syntactically related problem to function inversion is the so-called \textit{black-box function computation}: an algorithm tries to compute $f(x)$, for a randomly chosen $x$, using an advice of length $s$ on $f$, and by querying  $f$ on $q$ inputs other than $x$.   \citet{Yao82}  proved that $s\cdot q \ge n$, and presented a linear,  non-adaptive algorithm  that matches this lower bound. \Inote{see my edits above}

A much-researched special case of this problem is known as the \textit{index coding} problem \cite{baryossef2011indexcoding}, originally inspired by information distribution over networks. In this setting, a single party is in possession of a vector $f$, and broadcasts a short message $a$ such that $n$ different recipients may each recover a particular value of $f$, using the broadcast message and knowledge of certain other values of $f$, as determined by a \emph{knowledge graph}. The analogy to  non-adaptive black-box function computation  is obvious when considering  $a$ as the advice string, and the access to various values of $f$ as queries. While \citeauthor{Yao90}'s bound on the time/memory tradeoff  also holds  for the index coding problem,  other  lower bounds, some of which consider ``linear'' algorithms \cite{baryossef2011indexcoding,haviv2012linear,lubetzky2009nonlinear,golovnev2018minrank,alon2020minrank},  do not seem to be relevant for the function inversion problem.

\subsection*{Open Questions}
The main challenge remains   to gain a better understanding on the power of adaptive and non-adaptive function inverters. A more  specific challenge is  to generalize our bound on affine decoders and decision trees to affine operations over  arbitrary (large) fields.

\subsection*{Paper Organization}
A rather detailed  description of  our proof technique is given in  \cref{sec:technique}.  Basic notations, definitions and facts are given in \cref{sec:Preliminaries}, where we also prove several basic claims regarding   random function inversion. The bound on linear-advice inverters is  given in \cref{sec:LinearAdvice}, and the bounds on non-adaptive inverters are given in \cref{sec:nonAdaptInverters}.   \ifdefined \IsLLnCS Omitted proofs can be found in the full version of this paper \cite{ChawinHM20ECCC}.\fi

\section{Our Technique}\label{sec:technique}
In this section we provide a rather elaborate description of our proof technique. We start with the bound on linear-advice  inverters in \cref{sec:technique:LinearAdvice}, and then in   \cref{sec:technique:NonAdaptive}  describe the bounds for  non-adaptive inverters.

\subsection{Linear-advice Inverters}\label{sec:technique:LinearAdvice} 
Our lower bound for inverters  with linear advice (and its immediate generalization to additive-advice inverters) is proved via a reduction from \emph{set disjointness}, a classical problem in the study of two-party communication complexity. In the set disjointness problem,  two parties, Alice and Bob, receive two subsets, $\cX$ and $\cY \subseteq[n]$, respectively, and by  communicating with each other try to   determine whether $\cX\cap \cY= \emptyset$. The question is how many bits the parties have to  exchange in order to output the right answer with high probability. Given an inverter  with linear advice,  we use  it to construct a protocol that solves the set disjointness problem on \emph{all} inputs in $\cQ\eqdef\set{\cX,\cY\subseteq [n] \colon \size{\cX\cap \cY} \le 1}$ by exchanging $s+q\cdot \log n$ bits.  \citet{razborov1992distributional} proved that to answer    with  constant success probability on all  input pairs in  $\cQ$,   the parties have to exchange $\Omega(n)$ bits.   Hence,  the above reduction  implies the desired lower bound on the time/memory tradeoff of such  inverters.

Fix a $q$-query  $s$-advice inverter $\Inv\eqdef(\preprocessor, \decoder)$ with linear advice, and assume for simplicity that $\Inv$'s success probability is one.  The following observation immediately follows by definition:  let $a_f \eqdef \preprocessor(f)$ and $a_g \eqdef \preprocessor(g)$  be the advice strings for  some functions  $f$ and $g\in \FFam$, respectively.  The  linearity of  $\preprocessor$ yields that $a \eqdef  a_f+ a_g   = \preprocessor(f +g)$. That is,  $a$ is the advice for the function $f+ g$ (all additions  are over the appropriate groups).  Given this  observation, we use $\Inv$ to solve set disjointness as follows:  Alice and Bob (locally) convert  their input sets $\cX$ and $\cY$ to functions $f_\Ac$ and $f_\Bc$ respectively, such that for any $x\in \cX\cap \cY$ it holds that $f(x)\eqdef (f_\Ac+ f_\Bc)(x)=0$, and $f(x)$ is \emph{uniform} for  $x\notin \cX\cap \cY$. Alice then sends $a_\Ac \eqdef \preprocessor(f_\Ac)$ to Bob who uses it to compute $a\eqdef \preprocessor(f) = a_\Ac +   \preprocessor(f_\Bc)$.  Equipped with the advice $a$ and the help of Alice, Bob then emulates $\decoder(0,a)$ and finds $x\in f^{-1}(0)$, if such exists. Since $f$ is unlikely to map many elements outside of $\cX \cap \cY$ to $0$,  finding such $x$ is highly correlated with  $\cX \cap \cY \ne \emptyset$.   In more details, the set disjointness protocol is defined as follows. 
 \begin{protocol}[Set disjointness protocol $\Pi= (\Ac(\cX),\Bc(\cY))$]~
\begin{enumerate}
	\item   \Ac  samples $f_\Ac \in \FFam$ by letting $f_\Ac(i)\eqdef 
	\begin{cases}
	0 & i \in \cX\\
	\sim [n] & \text{otherwise}.
	\end{cases}$

	\item \Bc  samples $f_\Bc \in \FFam$ analogously, \wrt $\cY$.

	\item[$-$] Let  $f \eqdef f_\Ac+ f_\Bc$.

	\item \Ac sends $a_\Ac \eqdef \preprocessor(f_\Ac)$ to  $\Bc$, and  $\Bc$ sets $a \eqdef  a_\Ac + \preprocessor(f_\Bc)$. \footnote{\Inote{new} If the inverter is only assumed to have additive advice, this step is replaced with the parties interacting in the guaranteed protocol  for computing the advice for $f$ from the descriptions of  $f_\Ac$ and $f_\Bc$.}
	
	\item $\Bc$ emulates $\decoder^f(0,a)$ while answering each query $r$ that $\decoder$  makes  to $f$ as follows:
	\begin{enumerate}
		\item \Bc sends $r$ to \Ac.
		\item \Ac sends $w_\Ac \eqdef f_\Ac(r)$ back to \Bc.
		\item \Bc replies $w \eqdef w_\Ac  + f_\Bc(r)$ to $\decoder$ (as the value of $f(r)$).
	\end{enumerate} 

   \item[$-$]  Let $x$ be $\decoder$'s answer at the end of the above emulation.

\item The parties reject if $x \in \cX \cap \cY$ (using an additional $\Theta( \log n)$ bits to   find it out), and accept otherwise.

\end{enumerate}
\end{protocol}
The communication complexity of $\Pi$ is essentially $s+q\cdot \log n$. It is also clear that the parties accept  if $\cX\cap \cY=\emptyset$. For the complementary case, by construction, the intersection point of   $\cX\cap \cY$  is in $f^{-1}(0)$. Furthermore, since $f(i)$ is a random value for all $i\notin \cX\cap \cY$, with constant probability \emph{only} the intersection point is in $f^{-1}(0)$. Therefore,   the protocol is likely to answer correctly also in the  case that $\size{\cX\cap \cY}=1$.

\subsection{Non-adaptive Inverters}\label{sec:technique:NonAdaptive}	
We focus on inverters with an affine decoder, and  discuss the extension to  affine decision tree decoders in  \cref{sec:technique:NonAdaptive:DT}.  The proof follows by bounding the  success probability of \emph{zero-advice}  inverters---the preprocessing algorithm  outputs an empty string. In particular, we prove that  the success probability of such inverters is at most  $2^{-\Omega(n)}$. Thus,  by a union bound over all advice strings, in order to invert $\FFam$ with high probability, the advice string of a general (non-zero-advice)  inverter has to be of length $\Omega(n)$.\footnote{This first part of the proof is rather standard, \cf  \citet{akshima2020time}.} Let $\Inv\eqdef (\queries,\decoder)$ be a zero-advice $q$-query   non-adaptive inverter with an  affine decoder. Let $F$ be a random element of $\FFam$, and for $i\in [n]$, let   $Y_i$ be a randomly and independently selected element  of $[n]$. Let $X_i \eqdef \decoder(Y_i,F(\queries(Y_i)))$, \ie $\Inv$'s answer on challenge $Y_i$, and  let $Z_i$ be the indicator for $\set{F(X_j) = Y_j}$ for all $j\in [i]$, \ie the event that $\Inv$ answers the first $i$ challenges correctly.   We prove the bound by showing that for some $m\in \Theta(n)$ it holds that 
\begin{align}\label{eq:technique:NonAdaptive:1}
\pr{Z_m} \in 2^{-\Omega(m)}
\end{align}
Note that \cref{eq:technique:NonAdaptive:1}   bounds the probability that \Inv inverts  $m$ random elements drawn from $[n]$ (where some of them might have no preimage at all), whereas we are interested in bounding the probability that  \Inv inverts a \emph{random output}  of $F$. Yet,  since $F$ is a random function, its image  covers with very high probability a constant fraction of $[n]$, and thus  \cref{eq:technique:NonAdaptive:1} can be easily manipulated to derive that  \remove{\Nnote{when we move to distribution over $f(x)$ we pay $\alpha_{\tau,\delta}$, which can be constant. In the proof we actually do it after the union bound over all the possible advices. Maybe we want to say something about it?}}
\begin{align}\label{eq:technique:NonAdaptive:2}
\ppr{f\gets \FFam}{ \ppr{\stackrel{x\gets[n]}{\stackrel{y=f(x)}{v = \queries(f,y)}}}{\decoder(y,f(v)) \in f^{-1}(y)}\ge 1/2}< 2^{-\Omega(m)}= 2^{-\Omega(n)}
\end{align}
Hence,  in order to invert a random function with high probability, a non-zero-advice inverter has to use advice of length $\Omega(n)$.

We prove \cref{eq:technique:NonAdaptive:1} by showing that for every $i\in [m]$ it holds that 
\begin{align}\label{eq:technique:NonAdaptive:3}
\pr{Z_i \mid Z_{i-1}} < 3/5
\end{align}
That is, for small enough  $i$,  the algorithm $\Inv$ is likely to fail on inverting the \ith challenge, even when conditioned on the successful inversion of the  first $i-1$ challenges. We note that it is easy to bound $\pr{Z_i \mid Z_{i-1}}$ for \emph{zero}-query  inverters. The conditioning on  $Z_{i-1}$ roughly gives $\Theta(i)$ bits of information about $F$. Thus, this conditioning   gives at most one  bit of information about $F^{-1}(Y_i)$, and the inverter  does not have enough information to invert  $Y_i$. When moving to non-zero-queries inverters, however, the situation gets much more complicated. By making the right queries, that may depend on $Y_{i}$, the inverter can exploit this ``small'' amount of information to find the preimage of $Y_i$. This is what happens, for instance, in  the  adaptive inverter of   \citet{Hellman80}. Hence, for bounding $\pr{Z_i \mid Z_{i-1}}$,  we critically exploit the assumption that  $\Inv$ is non-adaptive and has an affine decoder. In particular, we  bound $\pr{Z_i \mid Z_{i-1}}$ by translating the event $Z_{i}$ into an  affine system  of equations and then use a few   observations about the structure of the above   system to derive the desired bound.  These equations will have the form   $M \times F = V$, viewing $F$ as a random vector in $[n]^n$, for $\mM \eqdef \stack{\mM^{i-1}}{\mM^{i}}$  and $V \eqdef\stack{V^{i-1}}{V^{i}} $, such that:
\begin{enumerate}
	
		\item $\mM^{i-1}$ is a deterministic function of $(X_{< i},Y_{<i})$ and $\mM^{i}$ is a deterministic function of $Y_i$, letting $X_{< i}$ stand for $(X_1,\ldots,X_{i-1})$   and likewise for $Y_{< i}$. \label{item:technique:NonAdaptive:1}
		
	\item The event $M^{i-1}\times F' = V^{i-1}$  is the event
	$\bigwedge_{j<i} \set{(F'(X_j) = Y_j) \land  (\decoder(Y_j,F'(\queries(Y_j))) = X_j)}$, for $F'$ being a uniform, and independent, element of $\FFam$.  \label{item:technique:NonAdaptive:2}
	
	(In particular,  $M^{i-1}\times F=V^{i-1}$ implies that $Z_{i-1}$ holds,  and  binds the value of $(X_{<i},Y_{<i})$ to  $V^{i-1}$.)
	
	\item The event $M^i\times F' = V^i$  is the event $\set{\decoder(Y_i,F'(\queries(Y_i))) = X_i}$.
	
	(In particular,  $M^{i}\times F=V^{i }$ binds the value of $X_i$ to  $V^i$.)
	
	 \label{item:technique:NonAdaptive:3}
\end{enumerate}

The above $\mM$ and $V$ are defined as follows: assume for ease of notation that  $\Inv$ has a \emph{linear}, and not affine,  decoder. That is,  for every $y \in [n]$  there exists a ($\numqry$-sparse) vector $\al_y \in \fld^n$ such that $ \inrprd{\al_y, F} = X_y$. By definition, for every $j < i$:

\begin{enumerate}
	\item $\inrprd{\al_{Y_j}, F}=  X_{j}$.
\end{enumerate}
Conditioning on $Z_{i-1}$ further   implies that  for every $j < i$:
\begin{enumerate}
	\item[2.]  $ F(X_j) = Y_j$.
\end{enumerate}
Let $\ell \eqdef 2i-2$, and let $\mM^{i-1} \in \matl$ be the  (random) matrix defined by $\mM^{i-1}_{2k-1} \eqdef \alpha_{Y_k}$ and $\mM^{i-1}_{2k} \eqdef e_{X_k}$, letting $e_j$ being the \textit{unit vector} $(0^{j-1},1,0^{n-j})$. Let $V^{i-1} \in \vecl$ be the  (random) vector defined by $V^{i-1}_{2k-1} \eqdef  X_k$ and $V^{i-1}_{2k} = Y_k$. By definition, the event  $Z_{i-1}$ is equivalent to the event    $\mM^{i-1} \times F = V^{i-1}$. The computation $\Inv$ makes on input  $Y_i$ can also be described by  the  linear equation $\inrprd{\al_{Y_i}, F}=  X_i$.  Let  $\mM\eqdef \stack{\mM^{i-1}}{\al_{Y_i}}$ and $V \eqdef \stack{V^{i-1}}{X_i}$. We make use of the following claims (see proofs in \cref{sec:prelim:Linear}).
\begin{definition}[Spanned unit vectors]
		For a matrix $\mA\in\fld^{a\times n}$, let $\Cuv(\mA) \eqdef \set{ {j \in [n]} \colon e_j \in \Span(\mA) }$, for    $\Span(\mA)$ being the (linear) span of $\mA$'s rows. 
\end{definition}
That is, $\Cuv(\mA)$ is the set of indices of all unit vectors spanned by $\mA$. It is clear that $\size{\Cuv(\mA)} \le \rank(\mA) \le \min\set{a, n}$. The following  claim states that for  $j \notin \Cuv(\mA) $, knowing the value of  $\mA\times F$ gives no information about  $F_j$. 
\begin{claim}\label{clm:tech:known_unit_vectors}
	Let  $\mA\in \fld^{a\times n}$ and   $v \in   \Image(\mA)$. Then for every  $j\in [n] \setminus \Cuv(\mA)$ and  $y\in  [n]$, it holds that $\ppr{f\gets [n]^n}{f_j = y \mid  \mA\times f = v} = 1/n$.
\end{claim}
The second claim roughly states that by concatenating a $c$-row matrix to a given matrix $\mA$, one does not increase the  spanned unit set of $\mA$ by more than $c$ elements.
\begin{claim}\label{claim:tech:existence_general_specific_spans}
For every $\mA \in \fld^{\ell\times n}$ there exists an $\ell$-size set $\cs_A \subseteq [n]$ such that the following holds: for every $\mB \in \fld^{c\times n}$
there exists a $c$-size set $\cS_B \subseteq [n]$
such that  $\Cuv\stack{\mA}{\mB} \subseteq \cs_A \cup \cS_B$.
\end{claim}

For bounding $\pr{Z_i \mid Z_{i-1}}$ using the above observations, we   write
 \begin{align}\label{eq:tech:split_by_eM}
 \pr{Z_i \mid Z_{i-1}} &= \pr{Z_i \land X_i \in \Cuv(\mM) \mid Z_{i-1}}
 + \pr{Z_i \land X_i \notin \Cuv(\mM) \mid Z_{i-1}}
 \end{align}
 and finish the proof   by separately bounding the two terms of the above equation. Let $H \eqdef (X_i,Y_{\le i},\mM,V)$. We first note that 
 \begin{align}
 \lefteqn{\pr{Z_i \land X_i \notin \Cuv(\mM) \mid Z_{i-1}} \le  \pr{Z_i \mid  X_i \notin \Cuv(\mM), Z_{i-1}}}\\
 &=  \ex{ (x_i,y_{\le i},m,v) \gets H\mid X_i \notin \Cuv(\mM), Z_{i-1}}{\pr{F(x_i) = y_i \mid m\times F = v,Y_{\le i} = y_{\le i}}}\nonumber\\
 &=  \ex{ (x_i,y_{\le i},m,v) \gets H\mid X_i \notin \Cuv(\mM), Z_{i-1}}{\pr{F(x_i) = y_i \mid m\times F = v}}\nonumber\\
 &= 1/n.\nonumber
 \end{align} \newcommand{\vd}{T}
The first equality holds by definition of $Z_{i-1}$,  the second equality since $F$ is independent of $Y$, and the last one follows by \cref{clm:tech:known_unit_vectors}.  For bounding the left-hand term of \cref{eq:tech:split_by_eM}, let $\cs$ and  $\vd$ be the $\ell$-size set and the index guaranteed  by \cref{claim:tech:existence_general_specific_spans} for the matrices $\mM^{i-1}$ and vector $\al_{Y_i}$, respectively. Compute,

 \begin{align}
\pr{Z_i \land X_i \in \Cuv(\mM) \mid  Z_{i-1}}&\le \pr{Y_i \in F(\Cuv(\mM)) \mid  Z_{i-1}}\\
 &\le \pr{Y_i \in F(\cs \cup \set{\vd}) \mid  Z_{i-1}}\nonumber\\
 &\leq \pr{Y_i \in F(\cs)  \mid  Z_{i-1}} + \pr{Y_i = F(\vd)  \mid  Z_{i-1}}.\nonumber
 \end{align} 
 The second inequality is by \cref{claim:tech:existence_general_specific_spans}.
 Since  $F(\cs)$ is independent of $Y_i$, it holds that 
 \begin{align}
 \pr{Y_i \in F(\cs)  \mid  Z_{i-1}} \le \size{\cs}/n = \ell/n 
 \end{align}
 Bounding $ \pr{Y_i = F(\vd)  \mid  Z_{i-1}}$ is  more involved since  $\vd$ might depend on $Y_i$.\footnote{ Indeed, this  dependency between the  queries to $f$ and the value to invert  is exactly  what makes (efficient) inversion by adaptive  inverters  possible.} Yet  since $f$ is a random function, a simple counting argument yields that for any (fixed and independent of $f$) function $g$:
 \begin{align}\label{eq:not_many_good_indices_advanced}
\ppr{f \gets \FFam}{\ppr{y\gets [n]}{ y =  f(g(y))}\ge 1/2} \le  n^{-n/3}
 \end{align}
 Let $H\eqdef (X_{<i},Y_{<i})$, and for  $h = (x_{<i},y_{<i})\in \Supp(H)$ compute
 \begin{align}
 \lefteqn{\ppr{f\gets F|Z_{i-1},H=h}{\pr{Y_i = f(\vd)\mid H=h}\ge 1/2 }}\\
 &\le \frac 1{\pr{H=h,Z_{i-1}\mid Y_{< i } = y_{<i}}}\cdot \ppr{f\gets F\mid Y_{< i } = y_{<i}}{\pr{Y_i = F(\vd)\mid H=h}\ge 1/2 }\nonumber\\
 &= \frac 1{\pr{H=h,Z_{i-1}\mid Y_{< i } = y_{<i}}}\cdot \ppr{f\gets F}{\pr{Y_i = F(\vd)\mid H=h}\ge 1/2 }\nonumber\\
 &\le  \frac 1{\pr{H=h,Z_{i-1}\mid Y_{< i } = y_{<i}}}\cdot  n^{-n/3}\nonumber\\
  &\le   n^{n/4}\cdot  n^{-n/3} \in o(1).\nonumber
 \end{align}

 The first equality holds since $F$ is independent of $Y$. The second inequality holds by \cref{eq:not_many_good_indices_advanced}, noting that under the conditioning on $H=h$, the value of $T$ is a deterministic function of $Y_i$.
The third inequality holds since for not too big $i$,  $\pr{H=h,Z_{i-1}\mid Y_{< i } = y_{<i}}\ge n^{-n/4}$,
since this probabilistic event is essentially a system of linear equations over a randomly selected vector.
 \Nnote{ Was: The third inequality holds since for not too big $i$ we may approximate for the sake of simplicity that  $\pr{H=h,Z_{i-1}\mid Y_{< i } = y_{<i}}\ge n^{-n/4}$,
 due to this probabilistic event being essentially an array of linear equations over a randomly selected vector, so that each equation holds with probability about $1/n$.}
 Since the above holds for any $h$, we conclude that  
 $ \pr{Y_i = F(\vd)  \mid  Z_{i-1}} \le 1/2 + o(1)$. Putting it all together, yields that $\pr{Z_i \mid Z_{i-1}} < 1/n + \ell/n + 1/2 + o(1) < 3/5$, for not too large  $i$.

\subsubsection{Affine Decision Trees}\label{sec:technique:NonAdaptive:DT}
Similarly to the affine decoder case, we prove the theorem by bounding $\pr{Z_i \mid Z_{i-1}}$ for all ``not too large $i$''. Also in this case, we bound this probability by translating the conditioning on $Z_{i-1}$ into a system of affine  equations. In particular, we  would like to find  proper definitions for the matrix $\mM=\stack{\mM^{i-1}}{\mM^{i}}$ and  vector $V =\stack{V^{i-1}}{V^i} $, functions of $(X_{\le i},Y_{\le i})$, such that   conditions \ref{item:technique:NonAdaptive:1}--\ref{item:technique:NonAdaptive:3}  mentioned in the affine decoder case  hold.
   
We achieve  these  conditions by adding for each $j<i$  an equation for each of the linear computations done in the decision tree that computes $X_j$ from $Y_j$. The price is that  rather than having $\Theta(i)$ equations, we now have $\Theta(d\cdot i)$, for $d$ being the depth of the decision tree.  In order to have $\mM^{i}$  a deterministic function of $Y_i$ alone, we cannot simply  make
$\mM^i$ reflect the $d$ linear computations performed by the decoder, since each of these may depend on the results of previous computations, and thus depend on $F$. \remove{ $\mM^{i}$ a function of the the $d$ linear computations \emph{actually} used by the decoder  on  $X_i$ (as we did for affine decoders). This is  since what computation is done in a given node, is determined by the  \emph{results} of  the computations done at higher nodes. }   So rather, we have to add a row (\ie an equation) for each of the $q$ queries the decoder might use (queries that span all possible computations), which by definition also imply the dependency on $q$.  Taking these additional rows into account yields the desired bound.

\section{Preliminaries}\label{sec:Preliminaries}
\subsection{Notation}\label{sec:prelim:notation}
All logarithms considered here are in base two. We use calligraphic letters to denote sets, uppercase for random variables and probabilistic events, lowercase for functions and fixed values, and bold uppercase for matrices.  Let $[n] \eqdef \set{1,\ldots,n}$. Given a vector $v\in \Sigma^n$, let $v_i$ denote its \ith entry, let $v_{< i} \eqdef  v_{1,\ldots,i-1}$ and   let $v_{\le i} \eqdef v_{1,\ldots,i}$.  Let $\binom{[n]}{k}$ denote the set  of all subsets of $[n]$ of size $k$. The  vector $v$ is \emph{$q$-sparse} if it has no more than $q$ non-zero entries.

\paragraph{Functions.}
We naturally view functions from $[n]$ to $[m] $ as vectors in $[m]^n$, by letting   $f_i =  f(i)$.  For a finite ordered set $\cS\eqdef \set{ s_1,\ldots,s_k}$, let $f(\cS)\eqdef \set{ f(s_1), f(s_2), \ldots, f(s_k)}$.  Let  $f^{-1}(y)\eqdef \set{ x\in [n] \colon f(x)=y }$ and let $\Image(f) = \set{f(x)\colon x\in [n] }$. A function $f\colon\fld^n \to \fld$ , for a  field $\fld$  and $n\in \N$,   is \emph{affine} if there exist a vector $v\in \fld^n$ and a constant $\be \in \fld$ such that $f(x)=\inrprd{v,x} + \be$ for every $x \in \fld^n$, letting $\inrprd{v,x} \eqdef \sum v_i \cdot x_i$ (all operations are over  $\fld$).

\paragraph{Distributions and random variables.}
The support of a distribution $P$ over a finite set $\cS$ is defined by $\Supp(P) \eqdef \set{x\in \cS: P(x)>0}$. For a  set $\cS$, let $s\from \cS$ denote that $s$ is  uniformly drawn from $\cS$. 
For $\delta \in [0,1]$, let $h(\delta) \eqdef -\delta\log \delta - (1-\delta)\log(1-\delta)$, \ie the binary entropy function.

\subsection{Matrices and Linear Spaces}\label{sec:prelim:Linear}
We identify the elements of a finite  field of size $n$ with the elements of the set $[n]$, using some arbitrary, fixed, mapping. Let  $e_i$  denote the \ith  unit vector $e_j=(0^{j-1},1,0^{n-j})$.

For a matrix $\mA \in \fld^{a\times b}$, let  $\mA_i$ denote  the \ith row of $\mA$.  The span of $\mA$'s rows  is defined by  $\Span(\mA) \eqdef \set{ v \in \fld^b \colon \exists \delta_1,\ldots,\delta_a \in \fld \colon v = \sum_{i=1}^{a} \delta_i \mA_{i}}$. Let $ \Image(\mA) = \set{ v \in \fld^a  \colon \exists w\in \fld^b \suchthat \mA\times w = v}$, or equivalently, the image set of the function $f_\mA(w) \eqdef \mA\times w$.
We use the following well-known fact:
\begin{fact}\label{pre:solution_set_size}
	Let $\fld$ be a finite field of size $n$, let $\mA\in \fld^{a \times b}$, let $v\in \Image(\mA)$, and  let $\cF \subseteq \fld^b$ be the solution set of the system of equations $\mA \times F = v$. Then $ \size{\cF } = n^{b - \rank(\mA)}$.
\end{fact} 

We also make use  of the following less standard notion.
\begin{definition}[Spanned unit vectors]\label{def:CovUNitVectors}
	For a matrix $\mA\in\fld^{a\times b}$, let $\Cuv(\mA) \eqdef \set{ {j \in [b]} \colon e_j \in \Span(\mA) }$. 
	\Reviewer{note this duplicates 2.2 \Dnote{consider enforcing consistency in names} \Inote{please do}}
\end{definition}
That is, $\Cuv(\mA)$ is the indices of all unit vectors spanned by $\mA$. It is clear that $\size{\Cuv(\mA)} \le \rank(\mA) \le \min\set{a, b}$. It is also easy to see that  for any  $v\in  \Image(\mA)$, $\Cuv(\mA)$ holds those  entries that are \emph{common to all  solutions $w$ of the system $\mA\times w = v$}.~\footnote{That is, for every $i\in \Cuv(\mA)$, $w_i$ can be described as a linear combination of the entries of $v$, and thus $w_i$ is fixed by $v$.} The following  claim  states that for  $i\notin \Cuv(\mA)$, the number of solutions $w$ of the system $\mA\times w = v$ with $w_i=y$, is the same  for every $y$.

\begin{claim}\label{clm:known_unit_vectors}
	Let $\fld$ be a finite field of size $n$, let $\mA\in \fld^{a\times b}$  and $v \in   \Image(\mA)$. Then for every $i\in [n] \setminus \Cuv(\mA)$ and  $y\in  [n]$, it holds that $\ppr{f\gets [n]^b}{f_i = y \mid  \mA\times f = v} = 1/n$.
\end{claim}
\begin{proof}
	Let $\cF_{\mA,v} \eqdef \set{ f\in [n]^b \colon \mA\times f = v }$  be the set of solutions for the equation $\mA \times F = v$.
	Since, by assumption, $\mA\times F=v$ has a solution, by  \cref{pre:solution_set_size} it holds that $| \cF_{\mA,v} | = n^{b - \rank(\mA)}$.
	Next, let $\mA' \eqdef \stack{\mA}{e_i}, v' \eqdef \stack{v}{y}$,
	and $\cF_{\mA,v,i,y} \eqdef \set{ f\in [n]^b \colon \mA'\times f = v'}$ (\ie $\cF_{\mA,v,i,y}$ is the set of solutions for $\mA' \times F = v'$).
	Since, by assumption, $e_i \notin \Span(\mA)$, it holds that  $\mA' \times F = v'$ has a solution and $\size{\cF_{\mA,v,i,y}} = n^{b - \rank(\mA')} = n^{b - \rank(\mA)-1}$. We conclude that  $\ppr{f\gets [n]^b}{f_i = y \mid  \mA\times f = v} = \frac{ |\cF_{\mA,v,i,y}| }{ | \cF_{\mA,v}| }  = 1/n$.
\end{proof}

\newcommand{\qrymatrix}{ c }
\newcommand{\matdtemp}{ \fld^{\qrymatrix \times n} }
The following claim  states that adding a small number of rows to a given matrix $\mA$ does not increase the set $\Cuv(\mA)$ by much.
\begin{claim}\label{existence_general_specific_spans}
	For every $\mA \in \matl $ there exists an $\ell$-size set $\cS_\mA \subseteq [n]$
	such that the following holds:
	for any $ \mB \in \matdtemp$
	there exists a $\qrymatrix$-size set $\cS_\mB \subseteq [n]$
	for which $\Cuv\stack{\mA}{\mB} \subseteq \cS_\mA \cup \cS_\mB$.
\end{claim}

\begin{proof}
	Standard row operations performed on a matrix $\mM$ do not affect  $\Span(\mM)$, and thus do not affect $\Cuv(\mM)$. Therefore, we may assume that both $\mA$ and $\mB$ are  in row canonical form.\footnote{(1) all zero rows are at the bottom  (2) the first non-zero entry in each row is equal to $1$ (known as the ``leading 1'') (3)  the leading $1$ in each row appears strictly to the right of the leading $1$ in all the rows above it  (4) a column that contains a leading $1$ is zero in all other entries. It is a well-known  that a matrix can be reduced to row canonical form using  Gaussian elimination, and the set of columns containing a leading $1$ is unique.} For a matrix $\mM$ in row canonical form, let $\cL(\mM) \eqdef \{ i \in [n] \colon  \text{the \ith column of $\mM$ contains a leading 1} \}$.
	Let $\cS_\mA \eqdef \cL(\mA)$ and note that $ |\cS_\mA | = \rank(\mA) \le \ell$. Perform Gaussian elimination on $\stack{\mA}{\mB}$ to yield a matrix $\mE$ in row canonical form, and let $\cS_\mE \eqdef \cL(\mE)$. Note that $\cS_\mA \subseteq \cS_\mE$, since adding rows to a matrix may only expand the set of leading ones. Furthermore, $ \size{\cS_\mE } = \rank(\mE) \le \rank(\mA) + \qrymatrix$. Clearly, $\Cuv(\mE) \subseteq \cS_\mE$, and we can  write $\cS_\mE = \cS_\mA \cup \cS_\mB$, for $\cS_\mB \eqdef (\cS_\mE \setminus \cS_\mA)$. Finally, $\size{\cS_\mB} = \size{\cS_\mE} - \size{\cS_\mA} \le \rank(\mA) + \qrymatrix - \rank(\mA) = \qrymatrix$,	and the proof follows.
\end{proof}

\subsection{Random Functions}
Let $\fspace$ be the set of all functions from $[n]$ to $ [n]$. We make the following observations. 

\begin{claim} \label{not_many_good_indices}
	Let $\cS_1,\ldots,\cS_n \subseteq [n]$ be $\qrymatrix$-size sets, and for $f\in \fspace$ let $\cK_f\eqdef \set{ y\in [n] \colon y \in f(\cS_y) }$. Then for any $\mu \in [0,\half]$:
	\begin{equation*}
	\ppr{f\gets \fspace}{\size{\cK_f} \ge \mu n} \le 2^{2\lceil \mu n\rceil\log(1/\mu) + \lceil \mu n\rceil \log (\qrymatrix / n)}.
	\end{equation*}
\end{claim}
\begin{proof}
	For $\cT\eqdef\set{t_1,\ldots,t_{\lceil \mu n\rceil} } \subseteq [n]$,  let $\FFam_{\cT} \eqdef \set {f\in \fspace \colon \cT \subseteq \cK_f}$.  We make a rough over-counting for  the size of $\FFam_\cT$: 
	one can  describe $f \in \FFam_\cT$ by choosing  $x_i \in[n]$ for  each set $\cS_{t_i}$, and require that $f(x_i) = t_i$ (to ensure $t_i \in f(\cS_{t_i})$). There are at most $\qrymatrix^{\lceil \mu n\rceil}$ ways to perform these choices. There are no constraints on the remaining $n-\lceil \mu n\rceil$ values of $f$. Therefore 
	$\size{\FFam_{\cT}} \le \qrymatrix^{\lceil \mu n\rceil} \cdot n^{n-\lceil \mu n\rceil}$.
	This immediately implies that $\ppr{f\gets \fspace,\cT \from \binom{[n]}{\lceil \mu n\rceil}}{ \cT \subseteq \cK_f } \le \left(\frac{\qrymatrix}{n}\right)^{\lceil \mu n\rceil}$.
	We conclude that 
	\begin{align*}
	\ppr{f\gets \fspace}{ |\cK_f| \ge \mu n } &=  \pr{\exists \cT\subseteq \cK_f\colon |\cT|= \lceil \mu n\rceil}\\
	&\le \sum_{\cT \in {[n] \choose \lceil \mu n\rceil}} \ppr{f\gets \fspace}{ \cT \subseteq \cK_f}\le  \binom{n}{\lceil \mu n\rceil} \cdot \left( \frac{\qrymatrix}{n} \right)^{\lceil \mu n\rceil}\le 2^{2\lceil \mu n\rceil\log(1/\mu) + \lceil \mu n\rceil \log (\qrymatrix / n)}.
	\end{align*}
	The last inequality follows  \cref{pre:entropy_binomial_fact,pre:binary_entropy_bound}, and the fact that $\log(1/\mu) \ge \log(n/\lceil \mu n\rceil)$.
\end{proof}


\begin{claim}\label{not_many_good_indices_conditioned} Let $n\in \N$,  let $F \gets \FFam_n$ and let $W$ be an event (jointly distributed with $F$) of probability at least $ p$.  Let  $Y\gets [n]$ be independent of $F$ and $W$. Then for every $\qrymatrix$-size sets $\cS_1,\ldots,\cS_n \subseteq [n]$ and $\gamma \in [0,\half]$, it holds that
	\begin{align*}
	\pr{Y \in F(\cS_{Y}) \mid W }\le \gamma + 
	2^{2\lceil \gamma n\rceil\log(1/\gamma) + \lceil \gamma n\rceil \log (\qrymatrix / n) + \log (1/p)}.
	\end{align*}
\end{claim}

\begin{proof}
	Let $\cK_f \eqdef \set{ y\in [n] \colon y \in f(\cS_y) }$. For $\gamma \in [0,1]$, compute: 
	\begin{align}\label{eq:not_many_good_indices_conditioned:1}
	\lefteqn{\pr{ Y\in F(\cS_Y) \mid W}  = \pr{Y \in \cK_F \mid W }}\\
	& \le \pr{ \size{\cK_F} \ge \gamma n \mid W } \cdot \pr{Y \in \cK_F \mid W, \size{\cK_F} \ge \gamma n} + \pr{ \size{\cK_F} < \gamma n \mid W }\cdot\pr{Y \in \cK_F \mid W, \size{\cK_F} < \gamma n }\nonumber \\
	& \le \pr{ \size{\cK_F} \ge \gamma n \mid W } + \gamma .\nonumber 
	\end{align}
	The last  inequality holds since $Y$ is independent of $W$ and $F$. Since  $\pr{W} \ge p$,  it holds that:
	\begin{align}
	\pr{ |\cK_F| \ge \gamma n \mid W } & \le \frac{  \pr{ |\cK_F| \ge \gamma n }   }{	\pr{ W} } \le 2^{2\lceil \gamma n\rceil\log(1/\gamma) + \lceil \gamma n\rceil \log (\qrymatrix / n) + \log (1/p)}
	\end{align}
	The second inequality is by \cref{not_many_good_indices}. 
	We conclude that:
	\begin{align*}
	\pr{Y \in F(\cS_{Y}) \mid W }\le
	\gamma + 2^{2\lceil \gamma n\rceil\log(1/\gamma) + \lceil \gamma n\rceil \log (\qrymatrix / n) + \log (1/p)}.
	\end{align*}
\end{proof}

The next  claim bounds the probability that a random function compresses an image set.

\begin{claim}\label{pre:tau_delta_bound}
	For any $n\in \N$ and  $\tau, \delta \in [0,\half]$, it holds that 
	
	$\alpha_{ \tau, \delta }\eqdef \ppr{\getf }{ \exists \cX\subseteq [n]\colon \size{\cX} \ge \tau n \land  \size{f(\cX)} \le \delta n } \le 2^{n (h(\tau) + h(\delta)) + \lfloor\tau n\rfloor \log \delta}$. 
\end{claim}

\begin{proof}
	Compute:
	\begin{align*}
	\alpha_{ \tau, \delta }& = \ppr{\getf}{  \exists \cX,\cY\subseteq [n] \colon  \size{\cX} \ge \tau n \land \size{\cY} \le \delta n \land  f(\cX) \subseteq \cY }\\
	& \le \ppr{\getf}{  \exists \cX,\cY\subseteq [n] \colon  \size{\cX} = \lfloor\tau n\rfloor \land \size{\cY} = \lfloor \delta n\rfloor \land  f(\cX) \subseteq \cY }\\
	& \le \sum_{\cY \in {[n] \choose \lfloor \delta n\rfloor}} \sum_{\cX \in {[n] \choose \lfloor\tau n\rfloor}}\pr{f(\cX) \subseteq \cY}\nonumber
	\le \binom{n}{\lfloor \delta n\rfloor}\binom{n}{\lfloor\tau n\rfloor} \cdot  {\delta}^{\lfloor\tau n\rfloor}\le 2^{n (h(\tau) + h(\delta)) + \lfloor\tau n\rfloor \log \delta}\nonumber.
	\end{align*}
	The last inequality follows from \cref{pre:entropy_binomial_fact}, and since $h$ is monotone in $[0,\half]$.
\end{proof}

\newcommand{\prex}{\mathcal{P}_f(x)}
The last claim states that an algorithm that inverts $f(x)$ with good probability, is likely to  return $x$ itself.

\begin{claim}\label{claim:getting_the_correct_preimage}
	Let $\Inv$ be a function from $\fspace \times [n]$ to $[n]$ such that $\ppr{\stackrel{f\gets \fspace}{x\gets [n]}}{\Inv(f,f(x)) \in \pref{f(x)}} \ge \alpha$. Then, $\ppr{\stackrel{f\gets \fspace}{x\gets [n]}}{\Inv(f,f(x)) = x} \ge  \frac{\alpha^2}{8}$.
\end{claim} 

\begin{proof}	
	For $f\in \fspace$ let $\prex \eqdef \pref{f(x)} \setminus \set{x}$. We would like to provide a bound on the size of this set to ensure that $x$ is output with high probability. Compute
	\begin{align}\label{eq:linear_encoding:correct_x_bound}
	\ppr{ \substack{ \getf \\ x \from [n] } }{ \Inv(f,f(x)) = x }&= 
	\pr{ \Inv(f,f(x)) = x \land \Inv(f,f(x)) \in \pref{f(x)} }\\
	&\ge  \pr{ \Inv(f,f(x)) = x \mid \Inv(f,f(x)) \in \pref{f(x)} }\cdot \alpha. \nonumber
	\end{align}
	
	We now provide a lower bound for the left-hand term. For $d\geq1$ compute
	\begin{align}\label{eq:linear_encoding:correct_x_cond_success_bound}
	&\ppr{ \substack{ \getf \\ x \from [n] } }{ \Inv(f,f(x)) = x \mid \Inv(f,f(x)) \in \pref{f(x)} }\\ \nonumber
	&\ge 
	\pr{ \Inv(f,f(x)) = x \land |\prex| \le d \mid \Inv(f,f(x)) \in \pref{f(x)} }
	\\ \nonumber
	&=
	\pr{ \Inv(f,f(x)) = x \mid |\prex| \le d, \Inv(f,f(x)) \in \pref{f(x)} }
	\cdot
	\pr{ |\prex| \le d \mid \Inv(f,f(x)) \in \pref{f(x)} }
	\nonumber
	\\ \nonumber
	&\ge 
	\frac{1}{d+1} \cdot \pr{ |\prex| \le d \mid \Inv(f,f(x)) \in \pref{f(x)} }
	\nonumber
	\\ \nonumber
	&= \frac{1}{d+1} \left( 1 - \pr{ |\prex| > d \mid \Inv(f,f(x)) \in \pref{f(x)} } \right).
	\nonumber
	\end{align}

	By linearity of expectation, $\ex{\getf}{\size{\prex}} = \frac{n-1}{n} < 1$.
	Hence by  Markov's inequality, 
	
	$\ppr{\substack{ \getf \\ x \from [n] } }{ |\prex| > d } < 1/d$. It follows that 
	\begin{align}\label{eq:linear_encoding:preimages_upper_bound}
	\ppr{ \substack{ \getf \\ x \from [n] } }{ |\prex| > d \mid \Inv(f,f(x)) \in \pref{f(x)} } &\le \frac{ \pr{ |\prex| > d } }{ \pr{\Inv(f,f(x)) \in \pref{f(x)}} } \le \frac{1}{d\alpha}
	\end{align}
	
	Combining \cref{eq:linear_encoding:correct_x_cond_success_bound,eq:linear_encoding:preimages_upper_bound} yields that 
	\begin{align}\label{eq:linear_encoding:preimages_upper_bound:2}
	\ppr{\substack{ \getf \\ x \from [n] } }{ \Inv(f,f(x)) = x \mid \Inv(f,f(x)) \in \pref{f(x)} } \ge \frac{1}{d+1} \left( 1 - \frac{1}{d\alpha} \right)
	\end{align}
	Finally, by \cref{eq:linear_encoding:correct_x_bound,eq:linear_encoding:preimages_upper_bound:2} we conclude that
	\begin{align*}
	\ppr{\substack{ \getf \\ x \from [n] } }{ \Inv(f,f(x)) = x } \ge \frac{\alpha}{d+1} \left( 1 - \frac{1}{d\alpha} \right)  \ge \frac{\alpha}{2d} \left( 1 - \frac{1}{d\alpha} \right)= \frac{\alpha}{2d} - \frac{1}{2d^2}.
	\end{align*}
	Setting $d = \frac{2}{\alpha}$ yields that $\ppr{\substack{ \getf \\ x \from [n] } }{ \Inv(f,f(x)) = x } \ge \frac{a^2}{4} - \frac{\alpha^2}{8} = \frac{\alpha^2}{8}$.

\end{proof}

\subsection{Additional Inequalities}
We use the following easily-verifiable facts:
\begin{fact}\label{pre:log_bound}
	For $x \ge 1$: $\log x\geq 1-1/x$.
\end{fact}
\begin{fact}\label{pre:binary_entropy_bound}
	For $\delta \le 1/2$: $h(\delta) \le - 2 \delta \log \delta  $.
\end{fact}
We also use the following bound:
\begin{fact}[\cite{galvin2014three}]\label{pre:entropy_binomial_fact}
	$\binom{n}{k} \le 2^{n h(\frac{k}{n})}$.
\end{fact}

\newcommand{\prtcl}{\Pi}
\newcommand{\prtclA}{\pA}
\newcommand{\prtclB}{\pB}

\section{Linear-advice Inverters}\label{sec:LinearAdvice}
In this section we present  lower bounds on the time/memory tradeoff of adaptive function inverters with linear advice.  The extension to additive-advice  inverters is given in \cref{subsec:AdditiveAdvice}.

To simplify notation, the following definitions and results are stated \wrt a fixed $n\in \N$.  Let $\FFam$ be the set of all functions from $[n]$ to $[n]$.  All asymptotic notations (\eg $\Theta$) hide constant terms that are independent of $n$.   We start by formally defining adaptive function inverters.  
\begin{definition}[Adaptive inverters]\label{def:adaptiveInverterWithSuccess}
	An {\sf $s$-advice, $\numqry$-query adaptive inverter} is a  deterministic algorithm pair $\Inv\eqdef(\preprocessor, \decoder)$, where
	$\preprocessor:\FFam \to \str{s}$, and $\decoder^{(\cdot)}\colon [n] \times \str{s} \to [n]$ makes up to $q$ oracle queries. For $f\in \FFam$ and $y\in [n]$, let 
	$$\Inv(y;f) \eqdef  \decoder^f(y,\preprocessor(f)).$$
\end{definition}
That is,  $\preprocessor$ is  the \emph{preprocessing} algorithm that takes as input the function description and outputs a string of length $s$ that  we refer to as the \emph{advice} string.   The oracle-aided  $\decoder$ is the \emph{decoder} algorithm that performs the actual inversion action. It receives the element to invert $y$ and   the advice string,  and using $q$ (possibly adaptive) queries to $f$, attempts to output a preimage of $y$. Finally, $\Inv(y;f)$ is the  candidate preimage the algorithms of $\Inv$ produce for the element to invert  $y$ given the (restricted) access to $f$.   We define adaptive inverters with linear advice as follows, recalling that we may view $f\in \FFam$ as a vector $\in[n]^n$.

\begin{definition}[Linear preprocessing]\label{def:linearpreprocessing}
A  deterministic algorithm $\preprocessor\colon\FFam \to \str{s}$ is {\sf linear} if there exist an additive  group $\cG\subseteq \str{s}$ that contains $\preprocessor(\FFam)$, and an additive group $\cK$ of size $n$ such that for every $f_1,f_2 \in \FFam$ it holds that $\preprocessor(f_1 +_\cK f_2) = \preprocessor(f_1) +_\cG \preprocessor(f_2)$, letting $f_1+_\cK f_2 \eqdef ((f_1)_1+_\cK (f_2)_1,\dots, (f_1)_n+_\cK (f_2)_n)$.
\end{definition}

Below we omit the subscripts from $+_\cG$ and $ +_\cK$ when clear from the context.

We prove the bound for inverters with linear preprocessing by presenting a reduction from the well-known \emph{set disjointness} problem.
\begin{definition}[Set disjointness]\label{def:Setdisjointness}
	A protocol $\Pi=(\Ac,\Bc)$ {\sf solves   set disjointness  with error $\varepsilon$  over all inputs in $\cQ$} $\subseteq \set{ (\cX,\cY) \colon \cX, \cY \subseteq [\nat] } $, if for every $(\setx,\sety)\in \cQ$ 
	$$\ppr{\substack{r_\Ac \gets \zs, r_\Bc \gets \zs\\ r_p \gets \zs}}{(\Ac(\setx;r_\Ac),\Bc(\sety;r_\Bc))(r_p) = (\delta_{\setx,\setY},\delta_{\setx,\setY})} \ge 1- \varepsilon$$
	
	for $\delta_{\setx,\setY}$ being the indicator for $\setx \cap \sety = \emptyset$.
\end{definition}
Namely, except with  probability $\varepsilon$ over their private and public randomness, the two parties find out whether their input sets  intersect.   Set disjointness  is known to require large communication over the following set of inputs.

\remove{

\begin{definition}[Hard distribution for set disjointness]
	Let $\cQ^0 = \{ \cX, \cY \subset [n] \colon  |\cX| = |\cY| = \lfloor n/4 \rfloor,  \cX \cap \cY = \emptyset  \} $
	and let 
	$\cQ^1 = \{ \cX, \cY \subset [n] \colon  |\cX| = |\cY| = \lfloor n/4 \rfloor,  |\cX \cap \cY| = 1 \} $.
	Let $D^0$ and $D^1$ be the uniform distribution over $\cQ^0$ and $\cQ^1$ respectively,
	and let $D = \frac{3}{4}\cdot D^0 + \frac{1}{4}\cdot D^1$.
\end{definition}
\citet{razborov1992distributional} has shown that solving set disjointness $D_n$ with small error requires high communication complexity:
}
\begin{definition}[Communication complexity]
The {\sf communication complexity of a protocol $\Pi=(\Ac,\Bc)$}, denoted $CC(\Pi)$, is the maximal number of bits the parties exchange in an execution (over all possible  inputs and randomness).
\end{definition}

\begin{theorem}[Hardness set disjointness, \citet{razborov1992distributional}] \label{thm:razborov}
	Exists $\varepsilon > 0$ such that for every protocol $\Pi$ that solves set disjointness over {\sf all} inputs in $\cQ \eqdef \set{\cX,\cY \subseteq [n] \colon \size{\cX \cap \cY} \le 1 }$ with error $\varepsilon$, it holds that $CC(\Pi) \ge \Omega(n)$.~\footnote{\cite{razborov1992distributional} proved a stronger result:  there exists a distribution that fails all low communication protocols. For the sake of our argument, however, it is easier to work with the weaker statement of \cref{thm:razborov}.}
\end{theorem}

Our main result is the following reduction from  set disjointness to function inversion.
\begin{theorem}[From set disjointness to function inversion]\label{thm:linear_encoding:inverter_solves_set_disjointness}
	Assume exists an $s$-advice, $\numqry$-query linear-advice inversion algorithm with 	$\ppr{\stackrel{f\gets  \FFam}{x\gets [n]} }{\Inv(f(x);f) \in f^{-1}(f(x))} \ge \alpha$, and let $\cQ \eqdef \set{\cX,\cY \subseteq [n] \colon \size{\cX \cap \cY} \le 1 }$.
	 Then for every $\varepsilon>0$ there exists a protocol that solves set disjointness with (one-sided) error $\varepsilon$ and communication $O\left(\frac{\log(\varepsilon)}{\log(1-\alpha^2/8)} \cdot (s+q \log n)\right)$, on all  inputs in $\cQ$. 
\end{theorem}
Combining \cref{thm:razborov,thm:linear_encoding:inverter_solves_set_disjointness} yields the following bound on linear-advice inverters.
\begin{corollary}[\cref{thm:intro:LinearAdvice}, restated]\label{linear_encoding:corollary_cant_invert_well}
	Let $\Inv=(\preprocessor,\decoder)$ be an $s$-advice $\numqry$-query inversion algorithm with linear preprocessing such that $\ppr{\stackrel{f\gets  \FFam}{x\gets [n]} }{\Inv(f(x);f) \in f^{-1}(f(x))} \ge \alpha$. 
	Then $s + \numqry \log n \in \Omega(\alpha^2 \cdot n)$.
\end{corollary}

\begin{proof}[Proof of \cref{linear_encoding:corollary_cant_invert_well}]
	By \cref{thm:linear_encoding:inverter_solves_set_disjointness}, the existence of an $s$-advice, $\numqry$-query linear-advice inverter $\Inv$ with success probability $\ge \alpha$ implies  that set disjointness can be solved over $\cQ$, with error $\varepsilon > 0$ and communication complexity $O\left(\frac{\log(\varepsilon)}{\log(1-\alpha^2/8)} \cdot (s+q \log n)\right)$.
	Thus, \cref{thm:razborov} yields that $\frac{\log(\varepsilon)}{\log(1-\alpha^2/8)} \cdot (s+q \log n)\in \Omega(n)$.
	Since $\frac{\log(\varepsilon)}{\log(1-\alpha^2/8)} = \log(1/\varepsilon)\cdot\frac{1}{\log(1/(1-\alpha^2/8))}$, and since, 
	by \cref{pre:log_bound},  it holds that  $ \log(1/(1-\alpha^2/8)) \ge \al^2/8$, it follows that $s+\numqry \log n \in \Omega(\alpha^{2} \cdot n)$. 
\end{proof}

The rest of this section is devoted to proving \cref{thm:linear_encoding:inverter_solves_set_disjointness}. Fix an  $s$-advice, $\numqry$-query inverter $\Inv=(\preprocessor,\decoder)$  with linear preprocessing. We use $\Inv$ in  \cref{proto:set_intersection_reduction} to solve set disjointness. In the protocol below we identify a vector $v \in \str{n}$ with the set $\set{i : v_i =1}$.
\begin{protocol}[$\prtcl = ( \prtclA, \prtclB )$]~\label{proto:set_intersection_reduction}
	\item[$\prtclA$'s input:] $a\in \str{n}$.
	\item[$\prtclB$'s input:] $b\in \str{n}$.
	
	\item [Public randomness:] $d \in [n]$. 

	\item [Operation:] ~
	
		\begin{enumerate}

		\item $\prtclB$ chooses $y \gets [n]$.
		
		\item  $\prtclA$ constructs a function $f_\Ac: [n] \to [n]$ as follows:
			\begin{itemize}
				\item for $i$ such that $a_i=0$, it  samples $f_\Ac(i+d \mod n)$  uniformly at 	random.
			
				\item for $i$ such that $a_i=1$,  it sets $f_\Ac(i + d \mod n) = 0$.
			\end{itemize}
		\item  $\prtclB$ constructs a function $f_\Bc: [n] \to [n]$ as follows:
			\begin{itemize}
				\item for $i$ such that $b_i=0$, it samples $f_\Bc(i+d \mod n)$  uniformly at 	random.
			
				\item for $i$ such that $b_i=1$, it sets  $f_\Bc(i + d \mod n) = y$.
			\end{itemize}
		
		   \item[$-$] Let $f\eqdef f_\Ac + f_\Bc$.

			\item $\prtclA$ sends $\preprocessor(f_\Ac)$  to $\prtclB$.\label{proto:set_intersection_reduction:1}

			\item $\prtclB$ sets $c \eqdef \preprocessor(f_\Ac) +_\cG \preprocessor(f_\Bc)=\preprocessor(f)$.\label{proto:set_intersection_reduction:1:5}
		
			\item $\prtclB$ emulates $\decoder^{f}(y,c)$: whenever  $\decoder$ sends a query $r$ to $f$, algorithm $\prtclB$ forwards it to $\prtclA$, and feeds $f_\Ac(r) + f_\Bc(r)$ back into $\decoder$.\label{proto:set_intersection_reduction:2}
		
			\item[$-$] Let $x$ be  $\decoder$'s output in the above emulation, and let $i = x -d \mod n$.
		
			\item $\prtclB$ sends $(i,b_{i})$ to $\prtclA$. If $a_{i}=b_{i}=1$, algorithm $\prtclA$ outputs False  and informs $\prtclB$.\label{proto:set_intersection_reduction:accept}
	
		\item Otherwise, both parties output True.\label{proto:set_intersection_reduction:reject}
	\end{enumerate}
\end{protocol}
In the following we analyze the communication complexity and success probability of $\Pi$. 

\begin{claim}[$\Pi$'s communication complexity] \label{clm:set_intersection_reduction:CC}
It holds that ${CC(\Pi) \le s + 2\numqry (\log n + 1) + \log n + 3}$.
\end{claim}
\begin{proof}~
	\begin{enumerate}
		\item In \stepref{proto:set_intersection_reduction:1}, party $\prtclA$ sends $\preprocessor(f_\Ac)$  to $\prtclB$.
		
		\item  In \stepref{proto:set_intersection_reduction:2}, the parties exchange at most $2\log n +2$ bits for every query $\decoder$ makes.

		\item  In \stepref{proto:set_intersection_reduction:accept}, the parties exchange at most
		$\log n +3$ bits.
	\end{enumerate}
	Thus, the total communication is bounded by $s + 2\numqry (\log n + 1) + \log n + 3$.
\end{proof}

\begin{claim}[$\Pi$'s success probability]\label{clm:linear_encoding:one_round_success_probability}~
	\begin{enumerate}
		\item $\pr{(\prtclA(a), \prtclB(b)) = (\true,\true)} =1$ for every  $(a,b)\in \cQ^0 \eqdef \set{\cX,\cY \subseteq [n] \colon \size{\cX \cap \cY} = 0}$.
		
		\item  $\pr{(\prtclA(a), \prtclB(b)) = (\false,\false)} \ge \frac{\alpha^2}{8}$  for  every $(a,b)\in \cQ^1 \eqdef \set{\cX,\cY \subseteq [n] \colon \size{\cX \cap \cY} = 1}$.
	\end{enumerate}
\end{claim}
\begin{proof}
	By construction, it is clear that $\Pi$ always accepts (the parties output  True)  on inputs $(a,b)\in \cQ^0$.  Fix $(a,b)\in \cQ^1$, and let $Y, D, F, F_\Ac, F_\Bc$ and $I$ be the values of $y,d, f, f_\Ac,f_\Bc$ and $i$ respectively, in a random execution of $(\prtclA(a),\prtclB(b))$.  By construction,  $F(j) = F_\Ac(j) + F_\Bc(j)$ for all $j\in [n]$.  For $j$ not in the intersection,  either $F_\Ac(j) $ or $F_\Bc(j)$ is chosen uniformly at random by one of the parties, and therefore $F(j)$ is uniformly distributed and independent of all other outputs. For the intersection element $w$, it holds that $F(w)=y$, which  is uniform, and since there is exactly one intersection, is  independent from all other outputs.

	Let  $W \eqdef w +D \mod n$. Note that $W$ is uniformly distributed over $[n]$ and is independent of $F$. Also note that, by construction,  $Y = F(W)$. Therefore, $(F,W,Y)$ is distributed exactly as $(F, X, F(X))$ for $X\gets [n]$. Hence, the assumption on $\Inv$ yields that
	\begin{align*}
	\pr{ \Inv(Y ; F)  \in F^{-1}(Y)} \ge \alpha
	\end{align*}
	and by \cref{claim:getting_the_correct_preimage},
	\begin{align*}
	\pr{ \Inv(Y ; F) = W } \ge {\alpha^2}/{8}.
	\end{align*}
	Therefore, both parties output False with probability at least ${\alpha^2}/{8}$.
\end{proof}

\paragraph{Proving  \cref{thm:linear_encoding:inverter_solves_set_disjointness}} 
We now use  \cref{clm:set_intersection_reduction:CC,clm:linear_encoding:one_round_success_probability} to prove \cref{thm:linear_encoding:inverter_solves_set_disjointness}.

\begin{proof}[Proof of  \cref{thm:linear_encoding:inverter_solves_set_disjointness}]
Let $t= \ceil{ \frac{\log(\varepsilon)}{\log(1-\alpha^2/8)}}$, and consider the protocol $\Pi^t$, in which on input $(a,b)$ the parties interact in protocol $\Pi$ for $t$  times, and accept only if they do so in  \emph{all} iterations. By \cref{clm:set_intersection_reduction:CC,clm:linear_encoding:one_round_success_probability}, the communication complexity and success probability   of $\Pi^t$ in solving set disjointness over $\cQ$ match the theorem statement.
\end{proof}

\subsection{Additive-advice Inverters}\label{subsec:AdditiveAdvice}

The following result generalizes \cref{linear_encoding:corollary_cant_invert_well} by replacing the  restriction on the decoder (\eg linear and short output) with the ability to compute the advice string of $f_1 + f_2$  by a low-communication protocol over the inputs $f_1$ and $f_2$.

\begin{theorem}[Bound on additive-advice inverters]\label{thm:linear_encoding:additive_advice}
	Let $\Inv=(\preprocessor,\decoder)$ be an $\numqry$-query inversion algorithm such that $\ppr{\stackrel{f\gets  \FFam}{x\gets [n]} }{\Inv(f(x);f) \in f^{-1}(f(x))} \ge \alpha$. Assume exists a two-party protocol $(\Pc_1, \Pc_2)$ with communication complexity $k$ such that for every $f_1,f_2\in \FFam$,
	the output of $\Pc_2$ in $(\Pc_1(f_1), \Pc_2(f_2))$ equals to $\preprocessor(f_1+f_2)$  with probability at least $1-\gamma$  for some $\gamma \ge 0$, letting $f_1+f_2$ be according to  \cref{def:linearpreprocessing}. 
	Then ${k + \numqry \log n \in \Omega(\alpha^2 (1-\gamma)\cdot n)}$.
\end{theorem}

\begin{proof}
The proof follows almost the  exact same lines as that of \cref{thm:linear_encoding:inverter_solves_set_disjointness}, with the following changes: first, steps $4.$ and $5.$ in \cref{proto:set_intersection_reduction} are replaced by the parties $\Ac$ and $\Bc$ interacting in $(\Pc_1(f_\Ac), \Pc_2(f_\Bc))$,  resulting in $\Bc$ outputting  $\preprocessor(f_\Ac+f_\Bc)$  (thus, transmitting a total of $k+{2q(\log n+1)}+\log n + 3 \in O(k+q\log n)$ bits over the entire execution of the protocol). Second,  note that due to the constant failure probability of $(\Pc_1, \Pc_2)$ in  computing $\preprocessor(f_\Ac+f_\Bc)$, the success probability of each execution of the protocol is now lowered by a constant factor $(1-\gamma)$. This means that the rate of success when $\cX \cap \cY \neq \emptyset$ is now bounded from below by only $\al^2(1-\gamma)/8$ (rather than $\al^2/8$). The rest of the analysis is  identical to that of \cref{thm:linear_encoding:inverter_solves_set_disjointness}.

\end{proof}

\newcommand{\gx}{ \cG_\cX }
\newcommand{\gy}{ \cG_\cY }
\newcommand{\ilow}{{m^\ast}}
\newcommand{\maxithm}{m}
\newcommand{\mMH}{\mM^{i-1}}
\newcommand{\mMP}{\mM}
\newcommand{\vH}{V^{i-1}}
\newcommand{\vP}{V}

\section{Non-adaptive Inverters}\label{sec:nonAdaptInverters}
In this section we present  lower bounds on the time/memory tradeoff of non-adaptive function inverters.  In   \cref{sec:AffineDecoders}, we present a bound for non-adaptive affine decoders, and in  \cref{sec:DecisionTree} we extend this bound to  non-adaptive  affine decision trees.  To simplify notation, the following definitions and results are stated \wrt some fixed $n\in \N$, for which there exists a finite field of size $n$ which we  denote  by $\fld$.  Let $\FFam$ be the set of all functions from $[n]$ to $[n]$. All asymptotic notations (\eg $\Theta$) hide constant terms that are independent of $n$. We start by formally defining non-adaptive function inverters. 
\begin{definition}[Non-adaptive inverters]\label{def:nonAdaptInverters}
	An  {\sf $s$-advice $q$-query   non-adaptive  inverter} is a deterministic algorithm  triplet of the form $\Inv\eqdef \triplet$, where $
	\preprocessor\colon\FFam \to \str{s}$,
	$\queries\colon [n] \times \str{s} \to [n]^\numqry$, and 
	$\decoder \colon[n] \times \str{s} \times [n]^\numqry \to [n]$. For $f\in \FFam$ and $y\in [n]$, let 
	$$\Inv(y;f) \eqdef  \algOutputfa{y}{\preprocessor(f)}.$$
\end{definition}

That is,  $\preprocessor$ is  the \emph{preprocessing} algorithm.  It  takes the function description as input and outputs a string of length $s$, to which  we refer as the \emph{advice} string. In the case that $s=0$, we say that  $\Inv$ has  \emph{zero-advice}, and omit  $\preprocessor$ from the notation.  Algorithm $\queries$ is the \emph{query selection} algorithm.  It chooses  the queries according to the element to invert $y$ and the advice string, and  outputs $q$ indices, to which  we refer as the \emph{queries}.  Algorithm $\decoder$ is the \emph{decoder} algorithm that performs the actual inversion. It receives the element $y$,  the advice string and the function's answers to the (non-adaptive) queries selected by $\queries$ (the query indices themselves may be deduced from $y$ and the advice),  and attempts to output a preimage of $y$. Finally, $\Inv(y;f)$ is the  candidate preimage of $y$ produced by the algorithms of $\Inv$ given the (restricted) access to $f$.

\subsection{Affine Decoders}\label{sec:AffineDecoders}
In this section we present our bound for non-adaptive affine decoders,  defined as follows:

\begin{definition}[Affine decoder]\label{def:AffineDecoders}
A non-adaptive 	inverter $\Inv\eqdef \triplet$ has an {\sf   affine decoder}, if for every  $y\in [n]$ and  $a\in \zo^s$ there exists  a $\numqry$-sparse vector $\al_y^a\in\textbf{\fld}^n$ and a field element $\be_y^a\in\fld$, such that for every $f\in \FFam$:~~ $\decoder(y,a,f(\queries(y,a)))= \inrprd{\al_y^a, f}+ \be_y^a$.

\end{definition}
The following theorem bounds the probability, over a random function $f$, that a  non-adaptive inverter with an affine decoder inverts  a random output of $f$ with probability $\tau$.
\begin{theorem}\label{thm:AffineDecoders}
	Let $ \Inv = \triplet $ be an $s$-advice non-adaptive inverter with an affine decoder and  let  $\tau \in [0,1]$. Then for every  $\delta \in [0,1]$ and  $m \le  n/16$, it holds that 
	\begin{align*}
	\ppr{ f\gets  \FFam}{\ppr{\stackrel{x\gets [n]}{y=f(x)} }{\Inv(y;f) \in f^{-1}(y)} \ge \tau}
	\le
	 \al_{\tau,\delta} + 2^s\cdot  \delta^{-\maxithm} \cdot\prod_{j=1}^{\maxithm} \left( \frac{2 j }{n} + \max \set{ \sqrt[4]{ 1/  n}, \frac{4j }{ n  }}
	\right)
	\end{align*}

	for $\al_{\tau,\delta} \eqdef \ppr{f \from \FFam}{ \exists \tau n\text{-size set } \cX \subset[n]\colon \size{f(\cX)} \le \delta  n }$.
\end{theorem}
While it is not easy to see what is the best choice, per $\tau$, of the parameters   $\delta$ and $\maxithm$ above,    the following corollary (proven in \cref{sec:AffineDecoders:corollary_proof})  exemplifies the usability of \cref{thm:AffineDecoders}  by considering the consequences of  such a choice.

\begin{corollary}[\cref{thm:intro:AffineDecoders}, restated]\label{cor:AffineDecoders}
	Let $\Inv$ be as in \cref{thm:AffineDecoders}, let $\tau\geq 2 \cdot n^{-1/8}$ and assume  
	
  $\ppr{ f\gets  \FFam}{\ppr{\stackrel{x\gets [n]}{y=f(x)} }{\Inv(y;f) \in f^{-1}(y)} \ge \tau} \ge \nfrac 1 2$, then $ s \in \Omega( \tau^2 \cdot n)$.~\footnote{The constant $1/2$ lower bounding the probability is arbitrary.} 
\end{corollary}

 Our key step towards proving \cref{thm:AffineDecoders} is showing that even when conditioned  on the (unlikely) event that a zero-advice inverter successfully inverts $i-1$ random elements, the probability the inverter successfully inverts the next element is  still low. To formulate the above statement,  we define the following jointly distributed random variables: let $F$ be uniformly distributed over $\FFam$ and let  $Y= (Y_1,...,Y_n)$ be  a uniform vector over $[n]^n$. For  a zero-advice inverter,  we  define the following random variables (jointly distributed with $F$ and $Y$). 
\begin{notation}\label{notation:XZ}
For a  zero-advice inverter  $\InvB$, let  $X_i^\InvB \eqdef \InvB(Y_i;F)$, let $Z_i^\InvB$ be  the event  $\bigwedge_{j\in [i]} \set{F(X_j^\InvB)=Y_j}$, and let $X^\InvB = (X_1^\InvB,\ldots,X_n^\InvB)$.

\end{notation}
That is,  $X^\InvB_i$ is $\InvB$'s answers to the challenges $Y_i$,    and $Z^\InvB_i$ indicates whether $\InvB$ successfully answered each of the first $i$  challenges. 
Given the above notation, our main lemma is stated as follows:
\begin{lemma}\label{lem:AffineDecoders}
Let $\InvB$  be a zero-advice,  non-adaptive inverter with affine decoder  and let $Z^\InvB$ be as in \cref{notation:XZ}. Then for every $i\in [n]$ and $\mu \in [0,\half]$:
	\begin{align*}
	\pr{Z_{\maxi}^{\InvB} \mid Z_{\maxi-1}^{\InvB} } 	\le
	\frac{2i-1}{n} + \mu
	+
	2^{2\lceil \mu n\rceil\log(1/\mu) - \lceil \mu n\rceil \log (n) + (2i-2)\log n}.
\end{align*}
\end{lemma}
We prove \cref{lem:AffineDecoders} below, but first  use it to prove \cref{thm:AffineDecoders}.

\paragraph{Proving \cref{thm:AffineDecoders}.}
\cref{lem:AffineDecoders}  immediately yields a bound on the probability that \emph{$\InvB$, a zero-advice inverter, successfully inverts the first $i$ elements of $Y$}. For proving \cref{thm:AffineDecoders}, however, we need to   bound the probability that $\InvB$, and later on, an  inverter with non-zero advice, finds a preimage of a \emph{random output}  of $f$. Yet, the conversion between     these two measurements  is rather straightforward. Hereafter we assume $n\geq16$, as otherwise \cref{thm:AffineDecoders} is trivial, as $m = 0$.

\begin{proof}[Proof of \cref{thm:AffineDecoders}.]
Let $ \Inv = \triplet $, $\tau \in [0,1]$, $\delta \in [0,1]$ and  $m$ be as in the theorem statement. 	 Fix an advice string $a\in\str{s}$, and let $\Inv^a= (\queries^a,\decoder^a)$ denote the \emph{zero-advice} inverter obtained by hardcoding $a$ as the advice of $\Inv$ (\ie $\preprocessor^a(f)=a$ for every $f$). For $j\in [n]$, let $Z_j = Z^{\Inv^a}_j$ and  let $\mu_j \eqdef \max \set{ \sqrt[4]{1 / n}, \frac{4j }{ n }}$.
We start by showing that for every $j \leq n/16$ it holds that 
\begin{align}\label{eq:lem:AffineDecoders:0}
\pr{Z_j \mid Z_{j-1} }\le\frac{2 j }{n} +\mu_j 
\end{align}

Indeed, by \cref{lem:AffineDecoders}
\begin{align}\label{eq:lem:AffineDecoders:1}
\pr{Z_j \mid \Zto{j-1}}&\le \frac{2j-1}{n} + \mu_j + 	2^{\underbrace{2\lceil \mu_j n\rceil\log(1/\mu_j) - \lceil \mu_j n\rceil \log n + (2j-2)\log n}_\beta}
\end{align}

We write,
\begin{align}\label{eq:lem:AffineDecoders:2}
\beta = 
\underbrace{2\lceil \mu_j n\rceil\log(1/\mu_j) - \frac{\lceil \mu_j n\rceil}{2} \log n}_{\beta_1} + \underbrace{ \left(-\frac{\lceil \mu_j n\rceil}{2}\right) \log n + (2j-2)\log n}_{\beta_2}
\end{align}

Since
\begin{align*}
\beta_1\le \lceil \mu_j n\rceil \left( \log \frac{1}{\mu_j^2} - \log \sqrt{n} \right) = 
\lceil \mu_j n\rceil \left( \log \frac{ 1}{\mu_j^2\sqrt{n}} \right) \le 0
\end{align*}
and 
\begin{align*}
\beta_2
= \frac{-\lceil \mu_j n\rceil}{2}  \log n + 2j\log n  - 2\log n
\le \frac{-2j}{ n  } n \log n + 2j\log n  - 2\log n\leq
- 2\log n,
\end{align*}
we conclude   that  $\pr{Z_j \mid \Zto{j-1}} \le \frac{2j-1}{n} + \mu_j + 2^{-2\log n}\le\frac{2j}{n} + \mu_j$,  proving \cref{eq:lem:AffineDecoders:0}.

\cref{eq:lem:AffineDecoders:0} immediately  yields that  
\begin{align}\label{eq:top_proof:2}
	\pr{Z_{\maxithm}}&=  
	\prod_{j=1}^{\maxithm} \pr{Z_j \mid Z_{j-1} }
	 \le \prod_{j=1}^{\maxithm} \left( \frac{2 j }{n} +\mu_j \right)
\end{align}
We  use the above to produce a bound on the number of elements that $\Inv^a$ successfully inverts. Let $\gy^a(f) \eqdef \set{ y \in [n] \colon \Inv^a(y;f) \in f^{-1}(y)}$, and compute:
	\begin{align}\label{eq:AffineDecoders:1}
	\pr{Z_{\maxithm}} & = \ppr{ f \from \FFam } { \forall j \in[\maxithm] \colon Y_j \in \gy^a(f) } \\
	& \ge \ppr{ f \from \FFam }{ \forall j \in[\maxithm] \colon Y_j \in \gy^a(f) \bigwedge	|\gy^a(f)| \ge \delta n} \nonumber \\
	& =\ppr{ f \from \FFam }{\forall j \in[\maxithm] \colon Y_j \in \gy^a(f) \mid	|\gy^a(f)| \ge \delta n}\cdot \ppr{ f \from \FFam } { |\gy^a(f)| \ge \delta n }
	\nonumber \\
	& \ge \delta^\maxithm \cdot \ppr{ f \from \FFam }{ |\gy^a(f)| \ge \delta n }.
	\nonumber
	\end{align}
	
Combining \cref{eq:top_proof:2,eq:AffineDecoders:1} yields the following  bound on the number of images $\Inv^a$ successfully inverts:
\begin{align}\label{eq:AffineDecoders:2}
	\pr{ |\gy^a(f)| \ge \delta n } \le  \delta^{-\maxithm} \cdot \prod_{j=1}^{\maxithm} \left( \frac{2 j }{n} + \mu_j \right)
\end{align}
We now adapt the above bound to  (the non zero-advice ) $\Inv$.
Let $\gy(f) \eqdef \set{ y \in [n] \colon \Inv(y;f) \in f^{-1}(y)}$ and let $\gx(f) = f^{-1}(\gy(f) )$. By \cref{eq:AffineDecoders:2} and a union bound, 
\begin{align}\label{eq:AffineDecoders:3}
\ppr{ f \from \FFam }{ |\gy(f)| \ge \delta n } \le  2^s \cdot \delta^{-\maxithm}  \cdot \prod_{j=1}^{\maxithm} \left( \frac{2 j }{n} + \mu_j \right)
\end{align}
We conclude that 
\begin{align*}
\lefteqn{\ppr{ f\gets  \FFam}{\ppr{\stackrel{x\gets [n]}{y=f(x)} }{\Inv(y;f) \in f^{-1}(y)} \ge \tau} =\ppr{f \from \FFam }{ |\gx(f) |  \ge \tau n}}\\
 & = \ppr{ f \from \FFam }{ |\gx(f) | \ge \tau n \bigwedge |\gy(f)| < \delta n } +\ppr{ f \from \FFam }{ |\gx(f) | \ge \tau n \bigwedge |\gy(f)| \ge \delta n }\\
& \le\ppr{ f \from \FFam }{ |\gx(f) | \ge \tau n \bigwedge |\gy(f)| < \delta n } +
\ppr{ f \from \FFam }{ |\gy(f)| \ge \delta n } \\
& \le \al_{\tau,\delta}+ 2^s \cdot \delta^{-\maxithm} \cdot \prod_{j=1}^{\maxithm} \left( \frac{2 j}{n} + \mu_j \right).
\end{align*}
The second inequality follows by  the definition of $ \al_{\tau,\delta}$ and \cref{eq:AffineDecoders:3}.
\end{proof}

\subsubsection{Proving \cref{lem:AffineDecoders}}
In the rest of this section we prove \cref{lem:AffineDecoders}. Fix a zero-advice  non-adaptive inverter with an affine decoder $\InvB= (\queriesB, \decoderB)$, $i \in [n]$ and $\mu \in [0,\half]$. Let  $X\eqdef X^\InvB$ and, for  $j\in [n]$ let $Z_j\eqdef Z^\InvB_j$.  We start by  proving the following claim that bounds the probability in hand, assuming $X_i$, the inverter's answer, is  coming from a small linear space.  (Recall, from \cref{def:CovUNitVectors},  that  $\Cuv(\mM) = \set{ {j \in [m]} \colon e_j \in \Span(\mM) }$, where $e_j$ is the \jth unit vector in $\fld^n$.) 

\begin{claim}\label{clm:not_many_good_indices_advanced}
	Let  $\mA \in \matl$, let $v \in \Image(\mA)$,  let $\mB^1,\ldots,\mB^n \in \fld^{t\times n}$,  and, for $y\in [n]$, let $\mA^y\eqdef\stack{\mA}{\mB^{y}}$.  Then 
	\begin{align*}
	\pr{Y_i \in F(\Cuv(\mA^{Y_i})) \mid \mA\times F= v} \le\left(\frac{\ell}{n} + \mu \right)	+
	2^{2\lceil \mu n\rceil\log(1/\mu) + \lceil \mu n\rceil \log (t / n) + \ell\log n}.
	\end{align*}
\end{claim}

\begin{proof}
	 By \cref{existence_general_specific_spans} there exist an $\ell$-size set   $\cS \eqdef \cs_\mA$ and $t$-size sets $\set{\cS_k\eqdef \cS_{\mB^k}}_{k\in [n]}$ such that 
	\begin{align}
	\Cuv(\mA^y) \subseteq \cS \cup \cS_{y}
	\end{align}
	for every $y\in [n]$.  By \cref{pre:solution_set_size}, 
	
	\begin{align}
	\pr{\mA\times F= v}= \frac{n^{n-\rank(\mA)}}{n^n} \geq n^{-\ell}
	\end{align}
 Compute,
	\begin{align}
\pr{ Y_i \in F(\Cuv(\mA^{Y_i}))\mid  \mA\times F= v }
	 &\le\pr{ Y_i \in F(\cS \cup \cS_{Y_i})\mid  \mA\times F= v } \\
	&\le  \pr{ Y_i \in F(\cS)\mid \mA\times F= v } +		\pr{ Y_i \in F(\cS_{Y_i})\mid \mA\times F= v }\nonumber\\
	&\le  \frac{\ell}{n} + \pr{ Y_i \in F(\cS_{Y_i})\mid  \mA\times F= v }.\nonumber
	\end{align}
	The first inequality  holds since $\Cuv(\mA^{Y_i}) \subseteq \cS \cup \cS_{Y_i}$, and the last one since $\size{\cS} \le \ell$ and  $Y_i$ is independent of $F$.
	Applying \cref{not_many_good_indices_conditioned} \wrt  $p \eqdef n^{-\ell}$, $\gamma\eqdef\mu$,  $W \eqdef \set{\mA\times F=v}$,  $Y\eqdef Y_i$ and  the  sets $\cS_1, \ldots \cS_n$,  yields that 
	\begin{align}
	\pr{ Y_i \in F(\cS_{Y_i})\mid  \mA\times F=v } &\le 
	\mu + 2^{2\lceil \mu n\rceil\log(1/\mu) + \lceil \mu n\rceil \log (t / n) + \ell\log n}
	\end{align}
We conclude that  $	\pr{Y_i \in F(\Cuv(\mA(Y_i))) \mid \mA\times F= v} \le \frac{\ell}{n} + \mu + 2^{2\lceil \mu n\rceil\log(1/\mu) + \lceil \mu n\rceil \log (t / n) + \ell\log n}$.
\end{proof}

Given  the above  claim, we  prove \cref{lem:AffineDecoders} as follows.
\begin{proof}[Proof of \cref{lem:AffineDecoders}]  Since $\InvB$ has an affine decoder, for every $y \in [n]$ and $X \eqdef \InvB(y;F)$  there exist a $\numqry$-sparse vector $\al^y \in \fld^n$ and a field element $\be^y \in \fld$ such that $ \inrprd{\al^y, F} + \be^y = X$. Therefore, for every $j < i$:
	\newcounter{saveenum}
	\begin{enumerate}
		\item $\inrprd{\al^{Y_j}, F}= - \be^{Y_j} + X_{j}$.
	\end{enumerate}
Conditioning on $Z_{i-1}$ further   implies that  for every $j < i$:
	\begin{enumerate}
		\item[2.]  $ F(X_j) = Y_j$.
	\end{enumerate}
	Let $\ell \eqdef 2i-2$, and let $\mMH\in \matl$ be the (random) matrix defined, for every $j\in [i-1]$, by $\mMH_{2j-1} \eqdef \alpha^{Y_j}$ and $\mMH_{2j} \eqdef e_{X_j}$. Let $\vH \in \vecl$ be the (random) vector defined by
	$\vH_{2j-1} \eqdef -\be^{Y_j} + X_j$ and $\vH_{2j} = Y_j$. By definition, conditioned on $\Zto{i-1}$ it holds that $\mMH \times F = \vH$.
	This incorporates in a single equation all that is known about $F$ given $\Zto{i-1}$. To take into account the knowledge gained from the queries made while attempting to invert $Y_i$, we combine the above with $\alpha^{Y_i}$ and $\inrprd{\alpha^{Y_i},F}$, into the matrix 
	$\mMP \eqdef \stack{\mMH}{\alpha^{Y_i}} $ and vector $ \vP \eqdef \stack{\vH}{\inrprd{\alpha^{Y_i},F}} $.   By definition,  $\mMP \times F = \vP$. We write
	\begin{align}\label{eq:split_by_eM}
	\pr{Z_i \mid  \Zto{i-1}} &= \pr{Z_i \land X_i \in \Cuv(\mMP) \mid  \Zto{i-1} }
	+ \pr{Z_i \land X_i \notin \Cuv(\mMP) \mid  \Zto{i-1}}
	\end{align}
	and  prove the lemma by separately bounding the two terms of the above equation. Let $H \eqdef (Y_{< i}, \mMH,\vH)$, and  note that

\begin{align}\label{eq:not_many_good_indices_bound}
&\pr{Z_i \land X_i \in \Cuv(\mMP) \mid  \Zto{i-1} }  \le \pr{Y_i \in F(\Cuv(\mMP)) \mid   \Zto{i-1}} \\
&= \ex{h \gets H \mid \Zto{i-1}}{\pr{Y_i \in F(\Cuv(\mMP)) \mid H = h,Z_{i-1}}}\nonumber\\
&= \ex{h=(y_{< i}, m^{i-1},v^{i-1}) \gets H \mid \Zto{i-1}}{\pr{Y_i \in F\left(\Cuv\stack{m^{i-1}}{\alpha^{Y_i}}\right) \mid H = h,m^{i-1}\times F = v^{i-1}}}\nonumber\\
&= \ex{(y_{< i}, m^{i-1},v^{i-1}) \gets H \mid \Zto{i-1}}{\pr{Y_i \in F\left(\Cuv\stack{m^{i-1}}{\alpha^{Y_i}}\right) \mid Y_{< i}=y_{<i}, m^{i-1}\times F = v^{i-1}}}\nonumber\\
&= \ex{(y_{< i}, m^{i-1},v^{i-1}) \gets H \mid \Zto{i-1}}{\pr{Y_i \in F\left(\Cuv\stack{m^{i-1}}{\alpha^{Y_i}}\right) \mid  m^{i-1}\times F = v^{i-1}}}\nonumber\\
&\le \left( \frac{2i-2}{n} + \mu \right)
+ 2^{2\lceil \mu n\rceil\log(1/\mu) + \lceil \mu n\rceil \log (1 / n) + (2i-2)\log n}.\nonumber
\end{align}

The first inequality holds by the definition of $Z_i$. The second equality holds by the definition of $Z_{i-1}$. The third equality holds since   the event $\set{Y_{< i}=y_{<i}, m^{i-1}\times F = v^{i-1}}$ implies that $\set{\mM^{i-1}=m^{i-1},V^{i-1}=v^{i-1}}$. The last equality holds since $F$ is independent of $Y$, and the last inequality follows by  \cref{clm:not_many_good_indices_advanced} with respect to $\mA \eqdef m^{i-1}, v \eqdef v^{i-1} $, and $(\mB^1,\ldots,\mB^n) \eqdef (\al^1,\ldots,\al^n)$ (viewing  $\al^i$ as a matrix in $\fld^{1 \times n}$).

 For bounding the right-hand term of \cref{eq:split_by_eM}, let $H \eqdef (X_i,Y_{\le i}, \mM,V)$,  and compute

\begin{align}\label{eq:clm:known_unit_vectors_bound}
&\pr{Z_i \land X_i \notin \Cuv(\mMP) \mid  \Zto{i-1}} \le \pr{Z_i \mid  X_i \notin \Cuv(\mMP) , \Zto{i-1}} \\
&= \ex{h \gets H\mid X_i \notin \Cuv(\mMP) , \Zto{i-1}}{\pr{Z_i \mid H=h,Z_{i-1}}}\nonumber\\
&= \ex{h= (x_i,y_{\le i}, m,v) \gets H\mid X_i \notin \Cuv(\mMP) , \Zto{i-1}}{\pr{F(x_i) =y_i \mid H=h,m\times F = v}}\nonumber\\
&= \ex{(x_i,y_{\le i}, m,v) \gets H\mid X_i \notin \Cuv(\mMP) , \Zto{i-1}}{\pr{F(x_i) =y_i \mid Y_{\leq i}=y_{\le i},m\times F = v}}\nonumber\\
&= \ex{(x_i,y_{\le i}, m,v) \gets H\mid X_i \notin \Cuv(\mMP) , \Zto{i-1}}{\pr{F(x_i) =y_i \mid m\times F = v}}\nonumber\\
&= 1/n. \nonumber
\end{align}

 The second equality holds by the definition of $Z_{i-1}$. The third equality holds since   the event $\set{Y_{\le i}=y_{\le i}, m\times F = v}$ implies that $\set{\mM=m,V=v}$,  and $X_i$ is a function of $\vP$. The fourth equality holds since $F$ is independent from $Y$.  The last inequality follows by \cref{clm:known_unit_vectors}. Combining \cref{eq:split_by_eM,eq:not_many_good_indices_bound,eq:clm:known_unit_vectors_bound}, we conclude that
	\begin{align*}
	\pr{Z_i \mid  \Zto{i-1}}\le& \left( \frac{2i-2}{n} + \mu \right) + 2^{2\lceil \mu n\rceil\log(1/\mu) + \lceil \mu n\rceil \log (1 / n) + (2i-2)\log n} + 1/n \\
	&=  \frac{2i-1}{n} + \mu + 2^{2\lceil \mu n\rceil\log(1/\mu) - \lceil \mu n\rceil \log (n) + (2i-2)\log n}.
	\end{align*}
\end{proof}

\subsubsection{Proving \cref{cor:AffineDecoders}}\label{sec:AffineDecoders:corollary_proof}

\begin{proof}[Proof of \cref{cor:AffineDecoders}]
	We prove for  $\tau \le 0.16$, which clearly yields the same bound for larger values of $\tau$. Let  $\delta  \eqdef \tau^2$, and let $\alpha_{\tau,\delta}$ be  as in \cref{thm:AffineDecoders}. By \cref{pre:tau_delta_bound},
	\begin{align}
	\al_{\tau, \tau^2} \le  2^{n ( h(\tau) + h(\tau^2)) + \lfloor \tau n \rfloor \log \tau^2}\le & 2^{n ( h(\tau) + h(\tau^2) + \tau \log \tau^2)-\log \tau^2}
	\\& \le 2^{n \underbrace{( h(\tau) + h(\tau^2) + \tau \log \tau^2)}_\beta}\cdot \tau^{-2}
	\end{align}
	Since $\tau  \leq 0.16$, it holds that    $h(\tau) \le \frac{3}{2}\cdot \tau  \cdot \log (1/\tau)$. We also note that
	\begin{align}
	\beta   \le & \frac{3}{2} \cdot \tau \cdot \log \frac{1}{\tau}
	+ \frac{6}{2}\cdot \tau^2 \cdot \log \frac{1}{\tau} +2\tau \log \tau =   (3\tau-1/2) \cdot\tau\cdot \log \frac{1}{\tau}  \le \frac{-\log n}{200\sqrt[8]{n}}
	\end{align}
	The last inequality holds since, by  assumption, $0.16 \geq \tau \geq \frac{2}{\sqrt[8]{n}}$, noting that $\tau \cdot \log 1/\tau$ is monotonically increasing over $[0,0.16]$.  Given the above bound on $\al_{\tau, \tau^2} $ and the assumption on $\Inv$'s success probability,  \cref{thm:AffineDecoders} yields that for every  $m \le {n }/{16}$: 
	\begin{align}\label{corollary:eq:first_beta_bound}
	1/2 \le 2^{-(n^{7/8}\log n)/200}\cdot \tau^{-2}+2^s \delta^{-\maxithm} \prod_{j=1}^{\maxithm} \left( \frac{2 j }{n} + \max \left\{ \sqrt[4]{1 / n}, \frac{4j}{ n } \right\}\right)
	\end{align}
	Let $\maxithm \eqdef  {\delta n }/{16 }$.  Since $\tau \ge \frac{2}{\sqrt[8]{n}}$, for every $j\in [m]$ it holds that  $\frac{2 j}{n} + \max \set{ \sqrt[4]{1 / n}, \frac{4j }{ n  }}\le \delta/2$.  Thus, by \cref{corollary:eq:first_beta_bound}, 
	\begin{align}
	1/2 \le 2^{-(n^{7/8}/200-1)\log n} + 2^{s}\cdot \delta^{-m} \cdot (\delta/2)^m \le  2^{-\log n} + 2^{s - \maxithm}
	\end{align}
	We conclude that $s \in \Omega(m)$ and  thus $s\in \Omega(\tau^2 \cdot n)$.
\end{proof}

\subsection{Affine Decision Trees}\label{sec:DecisionTree}
In this section we present lower bounds for non-adaptive affine decision trees. These  are formally defined as follows:

\newcommand{\Ep}{g}
\newcommand{\Np}{p}

\begin{definition}[Affine decision trees]\label{def:affineTrees}
An $n$-input {\sf affine decision tree} over $\fld$ is  a labeled, directed,  degree $\size{\fld}$    tree $\cT$. Each internal node $v$ of $\cT$ has label $\alpha_v\in \fld^n$,  each leaf $\ell$ of $\cT$ has label $o_\ell \in \fld$,   and the $\size{\fld}$ outgoing edges of every internal node are labeled by the  elements of $\fld$.  Let $\Gamma_\cT(v,\gamma)$ denote the (direct) child  of $v$ connected via the edge labeled by $\gamma$.  The {\sf node path} $\Np = (\Np_1,\ldots,\Np_{d+1})$ of $\cT$  on input $w \in \fld^n $ is defined by:
\begin{itemize}
	\item $\Np_1$ is the root of $\cT$.
	\item $\Np_{i+1}=\Gamma_\cT(\Np_i,\inrprd{\alpha_{\Np_i},w})$.
\end{itemize}
The {\sf edge path of $\cT$ on   $w$} is defined by  $(\inrprd{\alpha_{\Np_1},w},\cdots,\inrprd{\alpha_{\Np_{d}},w})$. 
Lastly, the {\sf output of $\cT$ on  $w$}, denoted $\cT(w)$,  is the value of $o_{\Np_{d+1}}$. 
\end{definition} 
Note that the edge path determines the computation path and output. Given the above,  affine decision tree decoders are defined as follows.

\begin{definition}[Affine decision tree decoder]\label{def:AffineTreesDecoders}
	An 	inversion algorithm $\Inv\eqdef \triplet$ has a  {\sf $d$-depth  affine decision tree decoder}, if for every  $y\in [n]$,   $a\in \zo^s$ and $v= \queries(y,a)$, there exists  an $n$-input, $d$-depth affine decision tree $\cT^{y,a}$ such that $\decoder(y, a, f(v))=\cT^{y,a}(f)$. 
\end{definition}
Note that such a decision tree may be of size $O(n^d)$. The following theorem bounds the probability, over a random function $f$, that a  non-adaptive inverter with  an affine decision tree decoder inverts  a random output of $f$ with probability $\tau$.
\begin{theorem}\label{thm:AffineTreeDecoders}
	Let $ \Inv$ be an $s$-advice, $(\numqry \le n/16)$-query,  non-adaptive inverter with a $d$-depth affine  decision tree decoder,  and let  $\tau \in [0,1]$. Then for every  $\delta \in [0,1]$ and  $m \le \frac{n \log (n/\numqry)}{4(d+1) \log n}$ it holds that 
	\begin{align*}
	\ppr{ f\gets  \FFam}{\ppr{\stackrel{x\gets [n]}{y=f(x)} }{\Inv(y;f) \in f^{-1}(y)} \ge \tau}
	\le
	\al_{\tau,\delta} + 2^s\cdot  \delta^{-\maxithm} \prod_{j=1}^{\maxithm} \left( \frac{(d+1) j }{n} + \max \set{ \sqrt[4]{ \numqry/  n}, \frac{2(d+1)j \log n}{ n \log (n/\numqry) }}
	\right)
	\end{align*}
	for $\al_{\tau,\delta} \eqdef \ppr{\getf}{ \exists \tau n\text{-size set } \cX \subset[n]\colon |f(\cX)| \le \delta  n }$.
\end{theorem}

Comparing to the bound we derive on affine decoders (\cref{thm:AffineDecoders}), we are  paying above for the tree depth $d$, but also  for the number of queries $q$. In particular, we essentially  multiply each term of the  above product by the tree depth $d$, and  by $\frac{\log n}{\log(n/q)}$.  In addition, the theorem  only holds for smaller values of $m$.
The following corollary exemplifies the usability of \cref{thm:AffineTreeDecoders}  by considering the consequences of  two choices of parameters. 
\begin{corollary}[\cref{thm:intro:AffineTree}, restated]\label{cor:AffineTreeDecoders}
	Let $\Inv$ be as in \cref{thm:AffineTreeDecoders} and assume  
	
	$\ppr{ f\gets  \FFam}{\ppr{\stackrel{x\gets [n]}{y=f(x)} }{\Inv(y;f) \in f^{-1}(y)} \ge \tau} \ge \nfrac 1 2$, then the following holds:
	\begin{itemize}
		\item 	If $\numqry \le  n \cdot (\nfrac \tau 2)^8$, then $ s \in \Omega( \nfrac n d\cdot \nfrac {\tau^2} {\log n})$.
		\item 	If $\numqry \le  n^{1-\eps}$, then $ s \in \Omega(\nfrac n {d} \cdot \tau^2  \eps)$.
	\end{itemize}
\end{corollary}
\begin{proof}
Omitted, follows by \cref{thm:AffineTreeDecoders} using very similar lines to those used to derive  \cref{cor:AffineDecoders} from \cref{thm:AffineDecoders}.
\end{proof}

In the rest of this section we explain how to modify the proof of \cref{thm:AffineDecoders}, in order to derive the proof of \cref{thm:AffineTreeDecoders}.  Hereafter we assume $n \ge 16$, as otherwise the bound trivially holds.

 Let $F \from \FFam$ and let  $Y= (Y_1,...,Y_n)$ be  a uniform vector over $[n]^n$. For  a zero-advice affine decision tree inverter   $\InvB=(\queriesB, \decoderB)$, let $X^\InvB$ and $Z^\InvB$, jointly distributed with $F$ and $Y$, be according to \cref{notation:XZ}. The crux of the proof of \cref{thm:AffineTreeDecoders}  lies in the following lemma.
\begin{lemma}\label{lem:AffineTreeDecoders}
	Let $\InvB= (\queriesB, \decoderB)$  be a zero-advice, $(\numqry \le n/16)$-query, non-adaptive inverter with a $d$-depth  affine decision trees decoder, and let $Z^\InvB$ be as in \cref{notation:XZ}. Then for every $i\in [n]$ and $\mu\in[0,\half]$:
	\begin{align*}
	\pr{Z_{\maxi}^{\InvB} \mid Z_{\maxi-1}^{\InvB} } \le
	\frac{(d+1)i-1}{n} + \mu
	+
	2^{2\lceil \mu n\rceil\log(1/\mu) + \lceil \mu n\rceil \log (\numqry / n)  + (i-1)(d+1)\log n}.
	\end{align*}
\end{lemma}

\paragraph{Proving \cref{thm:AffineTreeDecoders}.}
\begin{proof}[Proof of \cref{thm:AffineTreeDecoders} ]
Omitted, follows \cref{lem:AffineTreeDecoders} using essentially the  same lines we used to derive   \cref{thm:AffineDecoders} from \cref{lem:AffineDecoders}. 
\end{proof}

\subsubsection{Proving \cref{lem:AffineTreeDecoders}}
 \begin{proof}[Proof of \cref{lem:AffineTreeDecoders}] 
 	Fix  $\InvB= (\queriesB, \decoderB)$,  $i \in [n]$ and $\mu \in [0,\half]$.  Let $\cT^y$ be the affine decision tree associated with the computation   of $\decoderB$ on input $y$,  let   $\Np^y(f)$ and $\Ep^{y}(f)$  be the node and edge paths, receptively,   of $\cT^{y}$ on $f$, and let  $\alpha^{y}(f)\eqdef (\alpha^{y}_1(f),\ldots, \alpha^{y}_d(f)),o^y(f)$ be the labels of  $\Np^y(f)$ according to $\cT^y$. For $j \in [n]$, let  $\Np^j = \Np^{Y_j}(F)$, $\Ep^j = \Ep^{Y_j}(F)$, $\alpha^j \eqdef  \alpha^{Y_j}(F)$ and  $o^j \eqdef  o^{Y_j}(F)$. Finally, let  $X\eqdef X^\InvB$ and $Z_j\eqdef Z^\InvB_j$.   By definition,  for every $j < i$:
	
	\begin{enumerate}
			\item	$\forall k\in [d]: \quad \inrprd{\alpha^j_k, F}=\Ep^k_j$, 

			\item$o^j=X_j$.
	\end{enumerate}
	Conditioning on $Z_{i-1}$ further   implies that  for every $j < i$:
	\begin{enumerate}
		\item[2.]  $ F(X_j) = Y_j$.
	\end{enumerate}
Let $\ell \eqdef (d+1)(i-1)$, and let $\mM^{i-1}\in \matl$ be the (random) matrix defined by $\mM^{i-1}_{(d+1)(j-1)+k}\eqdef \alpha^j_k$ for  $k\in [d]$,  and $\mM^{i-1}_{(d+1)j} \eqdef e_{X_j}$. Let $V^{i-1} \in \vecl$ be the (random) vector defined by $V^{i-1}_{(d+1)(j-1)+k} \eqdef \Ep^k_{j}$ and $V^{i-1}_{(d+1)j} \eqdef Y_j$. By definition, conditioned on $\Zto{i-1}$ it holds that $\mM^{i-1} \times F = V^{i-1}$. That is, the matrix $\mM^{i-1}$ contains also the internal computations done by the $i-1$ trees (and not only the final outcome of the each computation as in the affine decoder case). For  $y \in [n]$, let $Q^y \eqdef \queriesB(y) \in [n]^\numqry$  be the queries that $\InvB$ makes on input $y$, and	let $A^y \in [n]^\numqry$ be $F$'s answers to these queries. That is, for every $k \in [\numqry]$:
\begin{enumerate}
	\item[3.] $ F(Q^y_k )=A^y_k$.
\end{enumerate}
Let $\mE^y\in \matd$ be the matrix defined by $\mE^y_k \eqdef e_{Q^y_k}$. Let $\mM \eqdef \stack{\mM^{i-1}}{\mE^{Y_i}}$ and let $V \eqdef \stack{V^{i-1}}{A^{Y_i}}$. 
That is, we add to $\mM$ all queries made by $\InvB$ on input $Y_i$ (and not only the output of the computations made by $\decoderB$).  We write
\begin{align}\label{eq:split_by_eM_trees}
\pr{Z_i \mid  \Zto{i-1}} &= \pr{Z_i \land X_i \in \Cuv(\mM) \mid  \Zto{i-1} }
+ \pr{Z_i \land X_i \notin \Cuv(\mM) \mid  \Zto{i-1}}
\end{align}

	and, using the above notions,  prove the claim by separately bounding the two terms of the above equation. Let $H \eqdef (Y_{< i}, \mMH,\vH)$, and  note that

\begin{align}\label{eq:not_many_good_indices_bound_trees}
&\pr{Z_i \land X_i \in \Cuv(\mMP) \mid  \Zto{i-1} }  \le \pr{Y_i \in F(\Cuv(\mMP)) \mid   \Zto{i-1}} \\
&= \ex{h \gets H \mid \Zto{i-1}}{\pr{Y_i \in F(\Cuv(\mMP)) \mid H = h,Z_{i-1}}}\nonumber\\
&= \ex{h=(y_{< i}, m^{i-1},v^{i-1}) \gets H \mid \Zto{i-1}}{\pr{Y_i \in F\left(\Cuv\stack{m^{i-1}}{\mE^{Y_i}}\right) \mid H = h,m^{i-1}\times F = v^{i-1}}}\nonumber\\
&= \ex{(y_{< i}, m^{i-1},v^{i-1}) \gets H \mid \Zto{i-1}}{\pr{Y_i \in F\left(\Cuv\stack{m^{i-1}}{\mE^{Y_i}}\right) \mid Y_{< i}=y_{<i}, m^{i-1}\times F = v^{i-1}}}\nonumber\\
&= \ex{(y_{< i}, m^{i-1},v^{i-1}) \gets H \mid \Zto{i-1}}{\pr{Y_i \in F\left(\Cuv\stack{m^{i-1}}{\mE^{Y_i}}\right) \mid  m^{i-1}\times F = v^{i-1}}}\nonumber\\
&\le \left( \frac{(d+1)(i-1)}{n} + \mu \right)
+ 2^{2\lceil \mu n\rceil\log(1/\mu) + \lceil \mu n\rceil \log (\numqry / n) + (d+1)(i-1)\log n}.\nonumber
\end{align}

The first inequality holds by the definition of $Z_i$. The second equality holds by the definition of $Z_{i-1}$. The third equality holds since   the event $\set{Y_{< i}=y_{<i}, m^{i-1}\times F = v^{i-1}}$ implies that $\set{\mM^{i-1}=m^{i-1},V^{i-1}=v^{i-1}}$. The last equality holds since $F$ is independent of $Y$, and the last inequality follows by  \cref{clm:not_many_good_indices_advanced} with respect to $\mA \eqdef m^{i-1}, v \eqdef v^{i-1} $, and $(\mB^1,\ldots,\mB^n) \eqdef (\mE^1,\ldots,\mE^n)$.

For bounding the right-hand term of \cref{eq:split_by_eM}, let $H \eqdef (X_i,Y_{\le i}, \mM,V)$,  and compute

\begin{align}\label{eq:clm:known_unit_vectors_bound_trees}
&\pr{Z_i \land X_i \notin \Cuv(\mMP) \mid  \Zto{i-1}} \le \pr{Z_i \mid  X_i \notin \Cuv(\mMP) , \Zto{i-1}} \\
&= \ex{h \gets H\mid X_i \notin \Cuv(\mMP) , \Zto{i-1}}{\pr{Z_i \mid H=h,Z_{i-1}}}\nonumber\\
&= \ex{h= (x_i,y_{\le i}, m,v) \gets H\mid X_i \notin \Cuv(\mMP) , \Zto{i-1}}{\pr{F(x_i) =y_i \mid H=h,m\times F = v}}\nonumber\\
&= \ex{(x_i,y_{\le i}, m,v) \gets H\mid X_i \notin \Cuv(\mMP) , \Zto{i-1}}{\pr{F(x_i) =y_i \mid Y_{\leq i}=y_{\le i},m\times F = v}}\nonumber\\
&= \ex{(x_i,y_{\le i}, m,v) \gets H\mid X_i \notin \Cuv(\mMP) , \Zto{i-1}}{\pr{F(x_i) =y_i \mid m\times F = v}}\nonumber\\
&= 1/n. \nonumber
\end{align}

The second equality holds by the definition of $Z_{i-1}$. The third equality holds since   the event $\set{Y_{\le i}=y_{\le i}, m\times F = v}$ implies that $\set{\mM=m,V=v}$,  and $X_i$ is a function of $\vP$ (which contains all the answers to the queries of the decoder). The fourth equality holds since $F$ is independent from $Y$.  The last inequality follows by \cref{clm:known_unit_vectors}. Combining \cref{eq:split_by_eM_trees,eq:not_many_good_indices_bound_trees,eq:clm:known_unit_vectors_bound_trees}, we conclude that
\begin{align*}
\pr{Z_i \mid  \Zto{i-1}}\le& \left( \frac{(d+1)(i-1)}{n} + \mu \right) + 2^{2\lceil \mu n\rceil\log(1/\mu) + \lceil \mu n\rceil \log (\numqry / n) + (d+1)(i-1)\log n} + 1/n \\
&\leq \left( \frac{(d+1)i-1}{n} + \mu \right) +2^{2\lceil \mu n\rceil\log(1/\mu) + \lceil \mu n\rceil \log (\numqry / n) + (d+1)(i-1)\log n}.
\end{align*}
\end{proof}

\section*{Acknowledgment}
We are  thankful to Dmitry Kogan, Uri Meir and Alex Samorodnitsky for very useful discussions.
We also thank the anonymous reviewers for their comments.

\bibliographystyle{abbrvnat}
\bibliography{crypto}
\end{document}